\documentclass[11pt]{article}

\usepackage{amsthm}
\usepackage{graphicx} 
\usepackage{array} 

\usepackage{amsmath, amssymb, amsfonts, verbatim}
\usepackage{hyphenat,epsfig,subcaption,multirow}
\usepackage{nicefrac}
\usepackage{paralist}

\usepackage[font=footnotesize,labelfont=bf]{caption}

\usepackage[usenames,dvipsnames]{xcolor}
\usepackage[ruled]{algorithm2e}

\DeclareFontFamily{U}{mathx}{\hyphenchar\font45}
\DeclareFontShape{U}{mathx}{m}{n}{
      <5> <6> <7> <8> <9> <10>
      <10.95> <12> <14.4> <17.28> <20.74> <24.88>
      mathx10
      }{}
\DeclareSymbolFont{mathx}{U}{mathx}{m}{n}
\DeclareMathSymbol{\bigtimes}{1}{mathx}{"91}

\usepackage{tcolorbox}
\tcbuselibrary{skins,breakable}
\tcbset{enhanced jigsaw}

\usepackage[normalem]{ulem}
\usepackage[compact]{titlesec}

\definecolor{DarkRed}{rgb}{0.5,0.1,0.1}
\definecolor{DarkBlue}{rgb}{0.1,0.1,0.5}

\usepackage{nameref}
\definecolor{ForestGreen}{rgb}{0.1333,0.5451,0.1333}
\definecolor{Red}{rgb}{0.9,0,0}
\usepackage[linktocpage=true,
	pagebackref=true,colorlinks,
	linkcolor=DarkRed,citecolor=ForestGreen,
	bookmarks,bookmarksopen,bookmarksnumbered]
	{hyperref}
\usepackage{bookmark}
\usepackage[noabbrev,nameinlink]{cleveref}
\crefname{property}{property}{Property}
\creflabelformat{property}{#2(#1)#3}
\crefname{equation}{eq}{Eq}
\creflabelformat{equation}{#2(#1)#3}

\usepackage{bm}
\usepackage{url}
\usepackage{xspace}
\usepackage[mathscr]{euscript}

\usepackage{tikz}
\usetikzlibrary{arrows}
\usetikzlibrary{arrows.meta}
\usetikzlibrary{shapes}
\usetikzlibrary{backgrounds}
\usetikzlibrary{positioning}
\usetikzlibrary{decorations.markings}
\usetikzlibrary{patterns}
\usetikzlibrary{calc}
\usetikzlibrary{fit}
\usetikzlibrary{snakes}
\tikzset{vertex/.style={circle, black, fill=Yellow, line width=1pt, draw, minimum width=8pt, minimum height=8pt, inner sep=0pt}}
	
\usepackage{mdframed}

\usepackage[noend]{algpseudocode}
\makeatletter
\def\BState{\State\hskip-\ALG@thistlm}
\makeatother

\usepackage{cite}
\usepackage{enumitem}

\usepackage[margin=1in]{geometry}

\usepackage{diagbox}
\usepackage{mathabx}

\newtheorem{theorem}{Theorem}
\newtheorem{lemma}{Lemma}[section]

\newtheorem{corollary}[lemma]{Corollary}
\newtheorem{claim}[lemma]{Claim}
\newtheorem{fact}[lemma]{Fact}

\newtheorem{definition}[lemma]{Definition}

\newtheorem*{claim*}{Claim}
\newtheorem*{theorem*}{Theorem}
\newtheorem*{proposition*}{Proposition}
\newtheorem*{lemma*}{Lemma}
\newtheorem*{problem*}{Problem}

\crefname{lemma}{Lemma}{Lemmas}
\crefname{claim}{Claim}{Claims}

\newtheorem{mdresult}{Result}
\newenvironment{result}{\begin{mdframed}[backgroundcolor=lightgray!40,topline=false,rightline=false,leftline=false,bottomline=false,innertopmargin=2pt, innerleftmargin=10pt]\begin{mdresult}}{\end{mdresult}\end{mdframed}}

\newtheorem*{remark*}{Remark}

\newtheoremstyle{restate}{}{}{\itshape}{}{\bfseries}{~(restated).}{.5em}{\thmnote{#3}}
\theoremstyle{restate}

\theoremstyle{definition}
\newtheorem{mdalg}{Algorithm}
\newenvironment{Algorithm}{\begin{tbox}\begin{mdalg}}{\end{mdalg}\end{tbox}}

\newtheorem{mddist}{Distribution}
\newenvironment{Distribution}{\begin{tbox}\begin{mddist}}{\end{mddist}\end{tbox}}

\newtheorem*{mdinvariant}{Example}

\allowdisplaybreaks

\DeclareMathOperator*{\argmax}{arg\,max}

\renewcommand{\qed}{\nobreak \ifvmode \relax \else
      \ifdim\lastskip<1.5em \hskip-\lastskip
      \hskip1.5em plus0em minus0.5em \fi \nobreak
      \vrule height0.75em width0.5em depth0.25em\fi}

\newcommand{\cc}{{\mathbf{CC}}}
\newcommand{\ic}{{\mathbf{IC}}}
\newcommand{\bD}{{\mathbf{D}}}
\newcommand{\cB}{{\mathcal{B}}}
\newcommand{\cX}{{\mathcal{X}}}
\newcommand{\cY}{{\mathcal{Y}}}
\newcommand{\rI}{{\rv{I}}}

\newcommand{\ur}{\textbf{UR}^\subset}
\newcommand{\dur}{\textbf{UR}^\subset_{\textrm{dec}}}
\newcommand{\mur}{\textbf{UR}^\subset_{\textrm{min}}}
\newcommand{\dmur}{\textbf{UR}^\subset_{\textrm{min,dec}}}
\newcommand{\mst}{\textbf{MST}}
\newcommand{\atpc}{\textbf{ATPC}}
\newcommand{\hintatpc}{\textbf{Hint-ATPC}}
\newcommand{\sender}{\texttt{sender}}
\newcommand{\receiver}{\texttt{receiver}}

\newcommand{\tvd}[2]{\ensuremath{\norm{#1 - #2}_{\mathrm{tvd}}}}
\newcommand{\Ot}{\ensuremath{\widetilde{O}}}
\newcommand{\eps}{\ensuremath{\varepsilon}}

\newcommand{\bracket}[1]{\left[#1\right]}
\newcommand{\paren}[1]{\ensuremath{\left(#1\right)}\xspace}
\newcommand{\card}[1]{\left\vert{#1}\right\vert}

\newcommand{\IN}{\ensuremath{\mathbb{N}}}

\newcommand{\norm}[1]{\ensuremath{\|#1\|}}

\newcommand{\ceil}[1]{{\left\lceil{#1}\right\rceil}}
\newcommand{\floor}[1]{{\left\lfloor{#1}\right\rfloor}}

\newcommand{\set}[1]{\ensuremath{\left\{ #1 \right\}}}
\newcommand{\poly}{{\mathrm{poly}}}

\DeclareMathOperator*{\Exp}{\ensuremath{{\mathbb{E}}}}
\DeclareMathOperator*{\Prob}{\ensuremath{\textnormal{Pr}}}
\renewcommand{\Pr}{\Prob}

\newcommand{\Ex}{\Exp}

\newenvironment{tbox}{\begin{tcolorbox}[
		enlarge top by=5pt,
		enlarge bottom by=5pt,
		 boxsep=2pt,
                  left=5pt,
                  right=7pt,
                  top=10pt,
                  arc=0pt,
                  boxrule=1pt,toprule=1pt,
                  colback=white
                  ]
	}
{\end{tcolorbox}}

\newcommand{\event}{\ensuremath{\mathcal{E}}}

\newcommand{\rv}[1]{\ensuremath{{\mathsf{#1}}}\xspace}
\newcommand{\rA}{\rv{A}}
\newcommand{\rB}{\rv{B}}
\newcommand{\rC}{\rv{C}}
\newcommand{\rD}{\rv{D}}

\newcommand{\rY}{\rv{Y}}
\newcommand{\supp}[1]{\ensuremath{\textnormal{\text{supp}}(#1)}}
\newcommand{\distribution}[1]{\ensuremath{\textnormal{dist}(#1)}\xspace}

\newcommand{\kl}[2]{\ensuremath{\mathbb{D}(#1~||~#2)}}
\newcommand{\II}{\ensuremath{\mathbb{I}}}
\newcommand{\HH}{\ensuremath{\mathbb{H}}}
\newcommand{\mi}[2]{\ensuremath{\def\mione{#1}\def\mitwo{#2}\mireal}}
\newcommand{\mireal}[1][]{
  \ifx\relax#1\relax%
    \II(\mione \,; \mitwo)%
  \else%
    \II(\mione \,; \mitwo\mid #1)%
  \fi
}
\newcommand{\en}[1]{\ensuremath{\HH(#1)}}

\newcommand{\itfacts}[1]{\Cref{fact:it-facts}-(\ref{part:#1})\xspace}

\newcommand{\cD}{\mathcal{D}}

\newcommand{\bias}{\ensuremath{\mathrm{bias}}}

\newcommand{\rX}{\rv{X}}

\newcommand{\rM}{\rv{M}}

\title{Optimal Multi-Pass Lower Bounds for MST in Dynamic Streams}
\author{Sepehr Assadi\footnote{(\texttt{sepehr@assadi.info})  Cheriton School of Computer Science, University of Waterloo, and Department of Computer
Science, Rutgers University.}\and
Gillat Kol\footnote{(\texttt{gillat.kol@gmail.com}) Department of Computer Science, Princeton University.}  \and 
Zhijun Zhang\footnote{(\texttt{zhijunz@princeton.edu}) Department of Computer Science, Princeton University.}  
}

\date{}

\begin{document}

\maketitle

\pagenumbering{roman}


\begin{abstract}

\medskip

The seminal work of Ahn, Guha, and McGregor in 2012 introduced the graph sketching technique and used it to present the first streaming algorithms
for various graph problems over \emph{dynamic} streams with both insertions and deletions of edges. This includes algorithms for cut sparsification, 
spanners, matchings, and minimum spanning trees (MSTs). These results have since been improved or generalized in various directions, leading to a vastly rich host of efficient algorithms for processing dynamic graph streams. 

\medskip

A curious omission from the list of improvements has been the MST problem. The best algorithm for this problem remains the original AGM algorithm that for every integer $p \geq 1$, uses $n^{1+O(1/p)}$ space in $p$ passes on $n$-vertex graphs, and thus
achieves the  desired {semi-streaming} space of $\Ot(n)$ at a relatively high cost of $O(\frac{\log{n}}{\log\log{n}})$ passes. On the other hand, no lower bounds beyond a folklore one-pass lower bound is known for this problem. 

\medskip

We provide a simple explanation for this lack of improvements: 
\begin{quote}
\emph{The {AGM algorithm for MSTs is optimal} for the entire range of its number of passes!} 
\end{quote}
We prove that 
even for the simplest \emph{decision} version of the problem---deciding whether the weight of MSTs is at least a given threshold or not--- 
any $p$-pass dynamic streaming algorithm requires $n^{1+\Omega(1/p)}$ space. This implies that semi-streaming algorithms do need $\Omega(\frac{\log{n}}{\log\log{n}})$ passes. 

\medskip

Our result relies on proving new \textbf{multi-round} communication complexity lower bounds for a variant of the \emph{universal relation} problem that has been instrumental in proving 
prior lower bounds for \emph{single-pass} dynamic streaming algorithms. The proof also involves proving new composition theorems in communication complexity, including majority lemmas and multi-party XOR lemmas, via information complexity approaches.

\end{abstract}

\clearpage

\setcounter{tocdepth}{3}
\tableofcontents
\clearpage

\pagenumbering{arabic}
\setcounter{page}{1}



%
%
%
%
%
%

\clearpage

\section{Introduction}\label{sec:intro}

In the \textbf{dynamic graph streaming} model, we have a (possibly edge-weighted) graph $G=(V,E)$ with vertices $V := \set{1,2,\ldots,n}$, whose edges and their weights are being defined by a sequence of insertions and deletions in a stream $\sigma := (\sigma_1,\sigma_2,\ldots,\sigma_N)$; here, $N$ is the length of the stream which is typically assumed to be $\poly(n)$. 
Each entry $\sigma_i$ is either of the form $(u_i,v_i,w_i,+)$ for $u_i,v_i \in V$ and $w_i \in \IN$ and is interpreted as the edge $(u_i,v_i)$ with weight $w(u_i,v_i) = w_i$ being inserted to $E$,
or $(u_i,v_i,w_i,-)$ which means the edge $(u_i,v_i)$ with the given weight $w_i$ is being deleted.  We are guaranteed that the stream does not delete an edge which is not inserted, does not insert an edge more than once before deleting it in the middle, and that the weight of a deleted edge matches its weight at the time of insertion\footnote{In particular, no ``partial updates'' to the edge weights are allowed and the stream needs to delete the edge ``fully'' first (and provide its weight) and then re-inserts it possibly with another  weight; see~\cite{ChenKL22} for more details on this.}. 
The goal is to make one or a few {sequential passes} over the stream $\sigma$, use a {limited memory}---ideally, $\Ot(n) := O(n \cdot \poly\!\log{\!(n)})$ bits, referred as the \textbf{semi-streaming space}---and compute an answer to the given problem on $G$ at the \emph{end} of the last pass.

Dynamic streams (not necessarily for graphs) have been studied extensively in the streaming literature  since the introduction of the model in~\cite{AlonMS96}, e.g., for statistical estimation problems~\cite{CharikarCF02} or geometric problems~\cite{FrahlingIS05}. 
However, despite the significant attention graph streams have received since their introduction in~\cite{FeigenbaumKMSZ05}, \emph{dynamic} graph streams were not studied for quite some time due to lack of any techniques
for addressing problems in this domain. 

This state-of-affairs was entirely changed by a seminal work of Ahn, Guha, and McGregor (henceforth, AGM)~\cite{AhnGM12a} who introduced the \emph{graph sketching} technique and used it to devise dynamic graph streaming algorithms
for several fundamental problems, including connectivity, minimum spanning trees, cut sparsifiers, and matchings. This immediately led to a flurry of results on dynamic graph streaming algorithms, \emph{all} using the graph sketching technique%
\footnote{The results in~\cite{LiNW14,AiHLW16} show that this is not a coincidence: any dynamic graph streaming algorithms that can handle triply-exponential long streams and doubly-exponential edge-multiplicities (in the middle of the stream), 
can be turned into a graph sketch. While these restrictions seem quite strong, almost all known graph streaming algorithms can handle such inputs as well. However, in this work, we will \emph{not} rely on this characterization.}, that either improved upon~\cite{AhnGM12a} or extended its results to various other problems; see, e.g.,~\cite{AhnGM12b,AhnGM13,KapralovLMMS14,BhattacharyaHNT15,ChitnisCHM15,McGregorTVV15,GuhaMT15,AhnCGMW15,BuryS15,AssadiKLY16,ChitnisCEHMMV16,HuangP16,FiltserKN21} and references therein.  

One of the very few problems that saw \emph{zero} improvement since~\cite{AhnGM12a} is the \emph{minimum spanning tree (MST)} problem.~\cite{AhnGM12a} designed a dynamic streaming algorithm that for every integer $p \geq 1$, with high probability, 
finds an MST of the input graph using $n^{1+O(1/p)}$ space and $p$ passes. Specifically, this leads to an $O(\frac{\log{n}}{\log\log{n}})$-pass semi-streaming algorithm. No better algorithms have been designed for this problem yet, despite the fact that in \emph{insertion-only} streams, a simple single-pass semi-streaming algorithm has already been known since~\cite{FeigenbaumKMSZ05}.

We provide a simple explanation for this lack of improvements: 
\begin{quote}
\emph{The {AGM algorithm for MSTs is optimal} for the entire range of its number of passes!} 
\end{quote}
Specifically, semi-streaming algorithms for MSTs require $\Omega(\frac{\log{n}}{\log\log{n}})$ passes. Beside settling the complexity of the fundamental MST problem in the semi-streaming model, 
this also constitutes one of the strongest separations between the power of insertion-only streams and dynamic graph streams; see, e.g.~\cite{AssadiKLY16,DarkK20} that prove such separations
only between single-pass algorithms (for the approximate matching problem).

\subsection{Our Contributions}

We now discuss our contributions in more detail. Our main result establishes the optimality of the MST algorithm of~\cite{AhnGM12a}. 

\begin{result}\label{res:main}
	For any integer $p = o(\frac{\log n}{\log\log n})$, any $p$-pass dynamic streaming algorithm on $n$-vertex graphs requires $\tilde{\Omega}(n^{1+\frac{1}{2p-1}})$ space to solve the minimum spanning tree problem with constant probability.
	The lower bound applies even if the edge weights and the length of the stream are both at most $O(n^2)$ and the algorithm only needs to decide whether the weight of minimum spanning trees is at least a given threshold.
\end{result}

Prior to our work, no lower bounds were known for the MST problem in dynamic streams beside a single-pass lower bound of $\Omega(n^2)$ space\footnote{To our knowledge, this lower bound appears to have been folklore and we do not know a reference for it.}. 
Another immediate corollary of our result is a strong limitation on the power of the graph sketching technique. While graph sketching has been extremely successful for problems such as 
cut- or spectral-sparsification~\cite{AhnGM12b,AhnGM13,KapralovLMMS14}, it appears to be quite weak for the MST problem, even when allowed ``many'' rounds of adaptive sketching.  

It is worth mentioning that our lower bound indeed only holds for \emph{exact} MSTs. For the relaxed version of the problem, wherein the goal is to obtain a $(1 + \eps)$-approximation instead,~\cite{AhnGM12a} already presents a 
single-pass semi-streaming algorithm. On the other hand, we prove our lower bound for exact MSTs for the algorithmically easiest \emph{decision} version of the problem: 
given a threshold at the beginning of the stream, decide whether the weight of MSTs is at least as large as this threshold or not. It is also worth mentioning that many problems admit provable separations between their search versus decision variants in the dynamic streaming model; see,~e.g.~\cite{AssadiKL17} for an example of a separation for finding approximate matchings versus estimating the size of the largest matchings via single-pass algorithms (or in~\cite{AssadiKL16} for the streaming set cover problem). 

\paragraph{Our techniques.}~\Cref{res:main} relies on proving a new \emph{multi-round} communication complexity lower bound for a \emph{non-standard composition} of a variant of the \textbf{Universal Relation (UR)} problem. 
UR has been instrumental in proving 
prior lower bounds for \emph{single-pass} dynamic streaming algorithms~\cite{JowhariST2011,KapralovNPWWY17,NelsonY19} (see also~\cite{Yu21}). In this problem, there is a universe $U$ of $m$ elements; Alice receives a set $A \subseteq U$ 
and Bob receives a proper subset $B \subset A$. The communication is only from Alice to Bob. Prior work has shown that in order for Bob to output \emph{any} element from $A \setminus B$, 
Alice needs to communicate $\Omega(\log^2{m})$ bits to succeed with constant probability~\cite{JowhariST2011} or $\Omega(\log^3{m})$ bits for high probability~\cite{KapralovNPWWY17}. 

We start by proving that any $r$-round protocol---wherein Alice and Bob can communicate back and forth at most $r$ times---for outputting the \emph{smallest} element in $A \setminus B$ (as opposed to outputting any one) requires $\Omega_r(m^{1/r})$ communication. We can then combine this with standard direct-sum arguments in communication complexity (see, e.g.~\cite{BarakBCR13}) to obtain that solving $m$ \emph{independent} copies of this problem
requires $\Omega_r(m^{1+1/r})$ communication. We then show how to reduce this to the problem of \emph{finding} MSTs in dynamic streams and prove a lower bound for the latter problem as well. This lower bound however does not extend to the decision problem (which is a common occurrence for other ``direct-sum UR-type'' reductions, e.g., in~\cite{NelsonY19} and~\cite{Yu21}).  

As we will explain in~\Cref{sec:overview}, to be able to extend the lower bound to the decision problem, the key ingredients used in our proof are: 
\begin{enumerate}[align=left]
	\item[\bf Direct sum with ``hint''.] At a high level, we will be dealing with a direct sum of a carefully defined variant of pointer chasing problems on trees.
		It differs from typical direct-sum arguments in that the reduction to MST demands knowing the sum of outputs of all copies, which {\em correlates} the copies.
		Our direct-sum result is obtained by directly carrying this extra bit of knowledge, named \emph{hint}, throughout the proof.
	\item[\bf Majority vs. XOR.] It turns out the most straightforward approach, which guesses the hint and conducts a typical direct-sum argument without the hint, can never work as it involves lower bounding majority computation of multiple copies with super low advantage.
		Simple coin toss examples will show that such a result is impossible.
		We work around this by a connection between majority computation with high advantage and XOR computation with low advantage.
		It enables us to utilize direct-sum results for XOR computation instead.
	\item[\bf Multi-party XOR lemma.] As known results are not strong enough for proving the optimal pass lower bound, we devise a multi-party XOR lemma, mimicking the $2$-party version of~\cite{Yu22}, that improves the dependence of communication in the number of rounds, while leading to a worse advantage decay.
		In particular, suppose each of $k$ pairs of $2$ parties are given $n/k$ instances of a boolean function $f$, and they want to jointly solve the $n$-fold XOR of all $n$ instances.
		We prove the following result which may be of independent interest.
\end{enumerate}

\begin{result}\label{res:xor}
	If any $r$-round, $2$-party protocol that solves $f$ with constant probability, requires $C$ communication, then any $r$-round, $2k$-party protocol that solves the $n$-fold XOR of $f$ with probability $\frac{1}{2}+(\frac{1}{2})^{\Omega(\frac{k}{r})}$, requires $\Omega(\frac{n}{k} \cdot (\frac{C}{r} - O(r)))$ communication.
\end{result}

The rest of this paper is organized as follows.
\Cref{sec:overview} provides a sketch of our proof in more detail.
Then we prove a suboptimal pass lower bound in~\Cref{sec:few} using the $2$-party XOR lemma of~\cite{Yu22}.
Our multi-party XOR lemma is presented in~\Cref{sec:optimal} and used to obtain the full version of our main result.

%
%
%
%
%



\section{Preliminaries}\label{sec:prelim}

\paragraph{Notation.}

For an integer $n \in \mathbb{N}$, $[n]$ is used as a shorthand for the set $\set{1,\ldots,n}$.
For a tuple $X = (X_1,\ldots,X_n)$, we write $X_{\le i} = (X_1,\ldots,X_i)$.
Similarly, we have $X_{\ge i}$ and $X_{< i},X_{> i}$.
We also use $X_{-i} = (X_1,\ldots,X_{i-1},X_{i+1},\ldots,X_n)$.
The XOR operation is denoted by $\oplus$.

Throughout this paper, sans-serif letters are reserved for random variables (e.g. $\rX$) while normal letters are used for realizations of the corresponding random variables (e.g. $x,X$).
For random variables $\rX,\rY$, we denote the \emph{Shannon entropy} of $\rX$ by $\en{\rX}$, the \emph{mutual information} between $\rX,\rY$ by $\mi{\rX}{\rY}$, the \emph{KL-divergence} between $\rX,\rY$ by $\kl{\rX}{\rY}$, and the \emph{total variation distance} between $\rX,\rY$ by $\tvd{\rX}{\rY}$.
\Cref{sec:info} provides necessary background on information theory, including the basic tools used in this paper.

\paragraph{Dynamic graph streaming.}

For a dynamic graph streaming problem, the input is a sequence of insertions and deletions of edges in an underlying graph, initially empty.
In every pass of an algorithm, it processes the operations, one at a time, in the given order.
At the end of the algorithm, it answers some query about the constructed graph resulting from all insertions and deletions.
Only the space requirement between operations is considered in this paper (i.e., unlimited memory is allowed while processing each operation).
We are interested in the problem $\mst_n$, which asks whether the weight of minimum spanning trees of an $n$-vertex graph is at least a given threshold.

\paragraph{Communication model.}

For the standard $2$-party communication model, we assume Alice sends the first message and the receiver of the last message returns the output.
Let $\cc(\pi)$ denote the \emph{communication complexity} of a protocol $\pi$, and $\cc^{(i)}(\pi)$ the communication complexity of the $i$-th round of $\pi$.
We also use $\ic(\pi)$ to denote the \emph{internal information cost} of $\pi$.
The \emph{distributional complexity} of $f$, denoted by $\bD^{(r)}_{\mu,\epsilon}(f)$, is defined as the infimum communication complexity of any $r$-round protocol solving $f$ with probability $\epsilon$ over $\mu$.

The multi-party communication model we use in this paper is formally defined as follows.
There are $2k$ parties named Alice $1,\ldots,k$ and Bob $1,\ldots,k$.
Each Alice has an input from $\cX$ and each Bob has an input from $\cY$.
There is a \emph{blackboard}, initially empty, visible to all parties.
The parties proceed in the \emph{circular} order of Alice $1,\ldots,k$ and Bob $1,\ldots,k$, starting with Alice $1$.
In one's turn, it computes a message given its input as well as the current blackboard, and posts the message to the blackboard.
At the end of the protocol, the last party returns an output (and does not post a message to the blackboard).
The communication complexity is defined as the length of the final blackboard.
The number of rounds is defined as the total number of times Alice $k$ and Bob $k$ post messages to the blackboard.
(So, e.g., a $1$-round protocol in general consists of Alice $1,\ldots,k$ and Bob $1,\ldots,k-1$ posting one message each, and Bob $k$ returning an output.)
In a randomized protocol, each party is allowed to use both public randomness, shared by all parties, and private randomness, known only to itself.
The goal is to compute a function $g$ over $\cX^k \times \cY^k$.
We similarly define the distributional complexity of $g$ in the $2k$-party model and denote it by $\bD^{(r),k}_{\mu,\epsilon}(g)$, where $\mu$ is a distribution over $\cX^k \times \cY^k$.
It can be verified that the multi-party model for $k=1$ coincides with the standard $2$-party model.
Moreover, $\bD^{(r),1}_{\mu,\epsilon}(\cdot) = \bD^{(r)}_{\mu,\epsilon}(\cdot)$.

In this paper, we are interested in the \emph{$k$-fold XOR} of a function $f : \cX \times \cY \to \set{0,1}$, defined as $f^{\oplus k}(x_1,\ldots,x_k,y_1,\ldots,y_k) = \bigoplus_{i \in [k]} f(x_i,y_i)$.
We also consider the \emph{$k$-fold majority}, denoted by $f^{\# k}$, which evaluates to $1$ if $f(x_i,y_i)=1$ for more than $\floor{k/2}$ indices $i \in [k]$, and $0$ otherwise.



\section{Technical Overview}\label{sec:overview}

This section serves as an outline of our proof.
As a starting point, in~\Cref{sec:search}, we first tackle the easier problem of proving a lower bound for the task of finding an MST solution, i.e., outputting the edges of an MST.
We then proceed to identify the primary challenges in extending our technique to give a lower bound for the algorithmically easier task of computing the weight of MSTs or even for the task of deciding whether it exceeds a specified threshold. In Section~\ref{sec:decision}, we discuss some of our initial attempts and their inherent limitations. Finally, we present the ultimate solution in Section~\ref{sec:multi-xor}.

\subsection{The Search Version}\label{sec:search}

\paragraph{Our hard instance.} We start by outlining our lower bound for the easier task of lower bounding the space complexity of steaming algorithms that output the edges of an MST.
To prove our lower bound, we design hard instances inspired by that of~\cite{NelsonY19,Yu21}, that were used to prove lower bounds for the Spanning Forest and Connectivity problems. See~\Cref{fig:search} for an illustration of our hard instances.
Our construction starts with a clique of size $n/2$.
Edges in the clique all have the minimum possible weight, say $0$.
Another $n/2$ vertices are added, one at a time, as follows.
For each non-clique vertex~$v$, it is randomly connected to some vertices in the clique, with distinct, positive edge weights. 
Later in the stream, we remove a proper subset of the edges incident on $v$.
Both the inserted and deleted edges follow some (non-uniform) hard distributions.
The concatenation of the clique edges (of weight $0$), followed by the edge insertions for all non-clique vertices, and then the edge deletions for all non-clique vertices, constitutes the entire stream.

\begin{figure}[!tbh]
	\centering

\begin{tikzpicture}

\tikzset{layer/.style={rectangle, rounded corners=5pt, draw, black, line width=1pt,  fill=black!10, inner sep=4pt}}
\tikzset{vertex/.style={circle, ForestGreen, fill=white, line width=2pt, draw, minimum width=8pt, minimum height=8pt, inner sep=0pt}}
\tikzset{choose/.style={rectangle, line width=1pt, rounded corners = 2pt, draw, minimum width=40pt, minimum height=16pt, fill=black!10}}

\begin{scope}[local bounding box=bb]
	\node[vertex] (w1) {};
	\foreach \i in {2,...,5}
	{
		\pgfmathtruncatemacro{\ip}{\i-1}
		\node[vertex] (w\i) [right=30pt of w\ip] {};
	}
\end{scope}

\begin{scope}[on background layer]
	\draw[layer, fill=gray!25] ($(bb.south west) - (10pt,10pt)$) rectangle ($(bb.north east) + (10pt,10pt)$);
\end{scope}

\node[vertex] (j1) [above=50pt of w1] {};

\foreach \i in {2,...,5}
{
	\pgfmathtruncatemacro{\ip}{\i-1}
	\node[vertex] (j\i) [right=30pt of j\ip] {};
}

\foreach \j in {2,...,5}
{
	\foreach \k in {1,...,5}
	{
		\draw[dashed, line width=1pt, gray, opacity=0.25] (j\j) -- (w\k);
	}
}

\draw[line width=2pt, red] (j1) -- (w1);
\draw[line width=2pt, blue] (j1) -- (w2);
\draw[line width=2pt, red] (j1) -- (w4);
\draw[line width=2pt, blue] (j1) -- (w5);

\end{tikzpicture}
	\caption{An illustration of hard instances for the search version of MST. Bottom vertices are fully connected. Each top vertex is connected to some bottom vertices via red edges (inserted and deleted) and blue edges (inserted but not deleted) -- to avoid clutter, only edges for the first vertex are drawn.}
	\label{fig:search}
\end{figure}
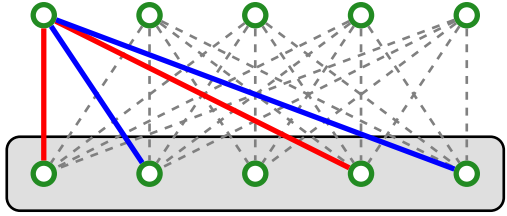

Observe that any MST of the constructed graph must have the following structure: a spanning tree connecting the clique, plus, for each non-clique vertex, the minimum weight edge that is not deleted connecting this vertex to the clique. As a consequence, the problem of finding an MST essentially reduces to the direct sum (i.e., solving multiple copies) of the following subproblem, which we denote by $\mur$: find the minimum element in the difference $A \setminus B$ of two sets $A,B$, where $B$ is promised to be a proper subset of $A$.

The problem $\mur$ can be viewed as an addition to the well-studied family of {\em Universal Relation problems}~\cite{KarchmerRW95}.
The work of \cite{NelsonY19} proves optimal lower bounds for Spanning Forest via one of its variants, $\ur$, in which it is sufficient to find {\em any} element in the difference $A \setminus B$, as opposed to finding the minimum element.
In particular, \cite{NelsonY19} use tight results from~\cite{KapralovNPWWY17} for the {\em one-way} communication complexity of $\ur$. However, this bound is only poly-logarithmic and therefore is too weak for our purposes. 
We prove that $\mur$ is hard even with multiple rounds of communication.
More specifically, we show that it admits an $r$ vs. $\Omega_r(m^{1/r})$ round-communication tradeoff, where $m$ is the size of the universe.
Given the canonical reduction from communication to streaming, this means any direct sum/product result for bounded-round two-party communication (e.g., \cite{JainPY12,BravermanRWY13}) suffices for lower bounding the search version of MST.

\paragraph{Augmented Tree Pointer Chasing.} We prove the round-communication tradeoff for $\mur$ by reduction from an ``augmented'' version of Pointer Chasing on trees\footnote{We note that $\mur$ is introduced here only for the purpose of illustration and to provide a better context. Our proofs in~\Cref{sec:few,sec:optimal} directly deal with the augmented version of Pointer Chasing with no reference to $\mur$. For completeness and since the lower bound for this problem may be of independent interest, we include its proof; see \Cref{cor:dmur}.}.
The starting point is the well-known {\em Augmented Index} problem~\cite{MiltersenNSW98}, in which Alice holds $x \in \set{0,1}^n$ while Bob is required to output $x_i$ given $i \in [n]$ and $x_{< i}$. 
It is an ``augmented'' version of Index in that Bob additionally knows $x_{< i}$, i.e., everything to the left of the pointer $i$.

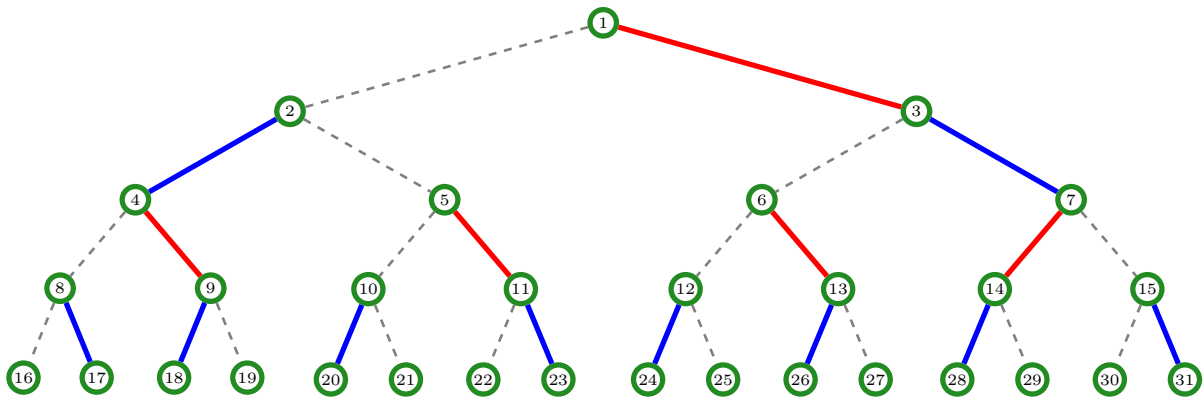
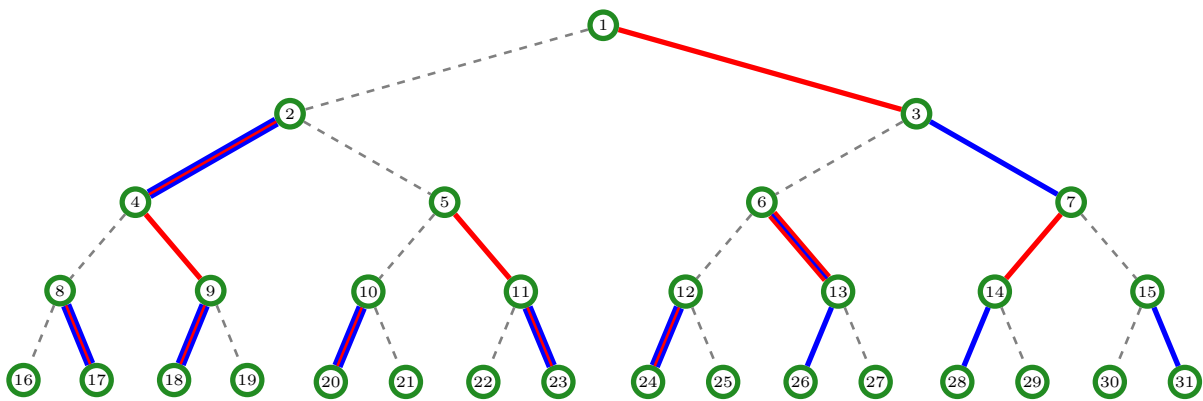
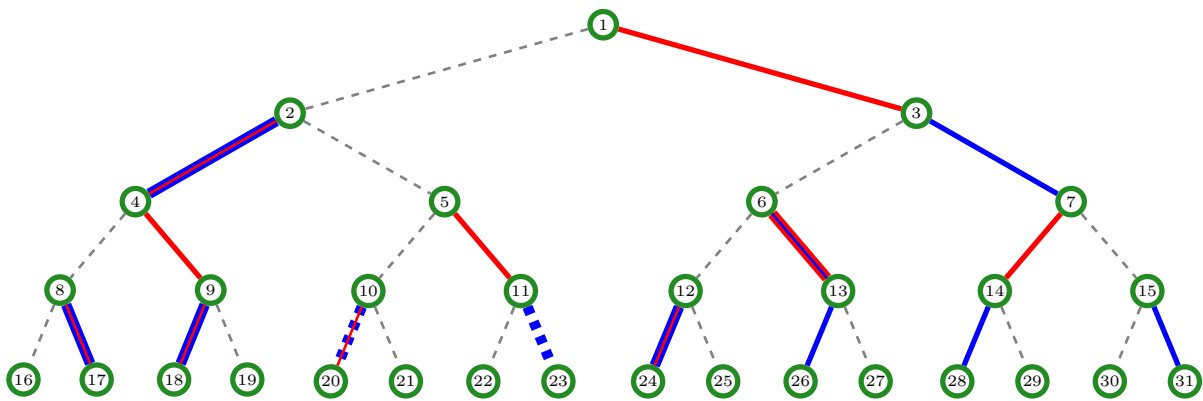
\begin{figure}[!tbh]
	\centering
	\begin{subfigure}{\textwidth}
		\centering

\begin{tikzpicture}

\tikzset{layer/.style={rectangle, rounded corners=5pt, draw, black, line width=1pt,  fill=black!10, inner sep=4pt}}
\tikzset{vertex/.style={circle, ForestGreen, fill=white, line width=2pt, draw, minimum width=10pt, minimum height=10pt, inner sep=1pt}}
\tikzset{choose/.style={rectangle, line width=1pt, rounded corners = 2pt, draw, minimum width=50pt, minimum height=16pt, fill=black!10}}

\node[vertex] (v1) [] {\tiny\color{black}1};

\node[vertex] (v2) [below left=25pt and 110pt of v1] {\tiny\color{black}2};
\node[vertex] (v3) [below right=25pt and 110pt of v1] {\tiny\color{black}3};

\draw[dashed, line width=1pt, gray, opacity=0.25] (v1) -- (v2);
\draw[line width=2pt, red] (v1) -- (v3);

\node[vertex] (v4) [below left=25pt and 50pt of v2] {\tiny\color{black}4};
\node[vertex] (v5) [below right=25pt and 50pt of v2] {\tiny\color{black}5};
\node[vertex] (v6) [below left=25pt and 50pt of v3] {\tiny\color{black}6};
\node[vertex] (v7) [below right=25pt and 50pt of v3] {\tiny\color{black}7};

\draw[line width=2pt, blue] (v2) -- (v4);
\draw[dashed, line width=1pt, gray, opacity=0.25] (v2) -- (v5);
\draw[dashed, line width=1pt, gray, opacity=0.25] (v3) -- (v6);
\draw[line width=2pt, blue] (v3) -- (v7);

\node[vertex] (v8) [below left=25pt and 20pt of v4] {\tiny\color{black}8};
\node[vertex] (v9) [below right=25pt and 20pt of v4] {\tiny\color{black}9};
\node[vertex] (v10) [below left=25pt and 20pt of v5] {\tiny\color{black}10};
\node[vertex] (v11) [below right=25pt and 20pt of v5] {\tiny\color{black}11};
\node[vertex] (v12) [below left=25pt and 20pt of v6] {\tiny\color{black}12};
\node[vertex] (v13) [below right=25pt and 20pt of v6] {\tiny\color{black}13};
\node[vertex] (v14) [below left=25pt and 20pt of v7] {\tiny\color{black}14};
\node[vertex] (v15) [below right=25pt and 20pt of v7] {\tiny\color{black}15};

\draw[dashed, line width=1pt, gray, opacity=0.25] (v4) -- (v8);
\draw[line width=2pt, red] (v4) -- (v9);
\draw[dashed, line width=1pt, gray, opacity=0.25] (v5) -- (v10);
\draw[line width=2pt, red] (v5) -- (v11);
\draw[dashed, line width=1pt, gray, opacity=0.25] (v6) -- (v12);
\draw[line width=2pt, red] (v6) -- (v13);
\draw[line width=2pt, red] (v7) -- (v14);
\draw[dashed, line width=1pt, gray, opacity=0.25] (v7) -- (v15);

\node[vertex] (v16) [below left=25pt and 5pt of v8] {\tiny\color{black}16};
\node[vertex] (v17) [below right=25pt and 5pt of v8] {\tiny\color{black}17};
\node[vertex] (v18) [below left=25pt and 5pt of v9] {\tiny\color{black}18};
\node[vertex] (v19) [below right=25pt and 5pt of v9] {\tiny\color{black}19};
\node[vertex] (v20) [below left=25pt and 5pt of v10] {\tiny\color{black}20};
\node[vertex] (v21) [below right=25pt and 5pt of v10] {\tiny\color{black}21};
\node[vertex] (v22) [below left=25pt and 5pt of v11] {\tiny\color{black}22};
\node[vertex] (v23) [below right=25pt and 5pt of v11] {\tiny\color{black}23};
\node[vertex] (v24) [below left=25pt and 5pt of v12] {\tiny\color{black}24};
\node[vertex] (v25) [below right=25pt and 5pt of v12] {\tiny\color{black}25};
\node[vertex] (v26) [below left=25pt and 5pt of v13] {\tiny\color{black}26};
\node[vertex] (v27) [below right=25pt and 5pt of v13] {\tiny\color{black}27};
\node[vertex] (v28) [below left=25pt and 5pt of v14] {\tiny\color{black}28};
\node[vertex] (v29) [below right=25pt and 5pt of v14] {\tiny\color{black}29};
\node[vertex] (v30) [below left=25pt and 5pt of v15] {\tiny\color{black}30};
\node[vertex] (v31) [below right=25pt and 5pt of v15] {\tiny\color{black}31};

\draw[dashed, line width=1pt, gray, opacity=0.25] (v8) -- (v16);
\draw[line width=2pt, blue] (v8) -- (v17);
\draw[line width=2pt, blue] (v9) -- (v18);
\draw[dashed, line width=1pt, gray, opacity=0.25] (v9) -- (v19);
\draw[line width=2pt, blue] (v10) -- (v20);
\draw[dashed, line width=1pt, gray, opacity=0.25] (v10) -- (v21);
\draw[dashed, line width=1pt, gray, opacity=0.25] (v11) -- (v22);
\draw[line width=2pt, blue] (v11) -- (v23);
\draw[line width=2pt, blue] (v12) -- (v24);
\draw[dashed, line width=1pt, gray, opacity=0.25] (v12) -- (v25);
\draw[line width=2pt, blue] (v13) -- (v26);
\draw[dashed, line width=1pt, gray, opacity=0.25] (v13) -- (v27);
\draw[line width=2pt, blue] (v14) -- (v28);
\draw[dashed, line width=1pt, gray, opacity=0.25] (v14) -- (v29);
\draw[dashed, line width=1pt, gray, opacity=0.25] (v15) -- (v30);
\draw[line width=2pt, blue] (v15) -- (v31);

\end{tikzpicture}
		\caption{A standard Tree Pointer Chasing instance.}
		\label{fig:atpc-raw}
	\end{subfigure}
	\vskip1cm
	\begin{subfigure}{\textwidth}
		\centering

\begin{tikzpicture}

\tikzset{layer/.style={rectangle, rounded corners=5pt, draw, black, line width=1pt,  fill=black!10, inner sep=4pt}}
\tikzset{vertex/.style={circle, ForestGreen, fill=white, line width=2pt, draw, minimum width=10pt, minimum height=10pt, inner sep=1pt}}
\tikzset{choose/.style={rectangle, line width=1pt, rounded corners=2pt, draw, minimum width=50pt, minimum height=16pt, fill=black!10}}

\node[vertex] (v1) [] {\tiny\color{black}1};

\node[vertex] (v2) [below left=25pt and 110pt of v1] {\tiny\color{black}2};
\node[vertex] (v3) [below right=25pt and 110pt of v1] {\tiny\color{black}3};

\draw[dashed, line width=1pt, gray, opacity=0.25] (v1) -- (v2);
\draw[line width=2pt, red] (v1) -- (v3);

\node[vertex] (v4) [below left=25pt and 50pt of v2] {\tiny\color{black}4};
\node[vertex] (v5) [below right=25pt and 50pt of v2] {\tiny\color{black}5};
\node[vertex] (v6) [below left=25pt and 50pt of v3] {\tiny\color{black}6};
\node[vertex] (v7) [below right=25pt and 50pt of v3] {\tiny\color{black}7};

\draw[line width=4pt, blue] (v2) -- (v4);
\draw[line width=1pt, red] (v2) -- (v4);
\draw[dashed, line width=1pt, gray, opacity=0.25] (v2) -- (v5);
\draw[dashed, line width=1pt, gray, opacity=0.25] (v3) -- (v6);
\draw[line width=2pt, blue] (v3) -- (v7);

\node[vertex] (v8) [below left=25pt and 20pt of v4] {\tiny\color{black}8};
\node[vertex] (v9) [below right=25pt and 20pt of v4] {\tiny\color{black}9};
\node[vertex] (v10) [below left=25pt and 20pt of v5] {\tiny\color{black}10};
\node[vertex] (v11) [below right=25pt and 20pt of v5] {\tiny\color{black}11};
\node[vertex] (v12) [below left=25pt and 20pt of v6] {\tiny\color{black}12};
\node[vertex] (v13) [below right=25pt and 20pt of v6] {\tiny\color{black}13};
\node[vertex] (v14) [below left=25pt and 20pt of v7] {\tiny\color{black}14};
\node[vertex] (v15) [below right=25pt and 20pt of v7] {\tiny\color{black}15};

\draw[dashed, line width=1pt, gray, opacity=0.25] (v4) -- (v8);
\draw[line width=2pt, red] (v4) -- (v9);
\draw[dashed, line width=1pt, gray, opacity=0.25] (v5) -- (v10);
\draw[line width=2pt, red] (v5) -- (v11);
\draw[dashed, line width=1pt, gray, opacity=0.25] (v6) -- (v12);
\draw[line width=4pt, red] (v6) -- (v13);
\draw[line width=1pt, blue] (v6) -- (v13);
\draw[line width=2pt, red] (v7) -- (v14);
\draw[dashed, line width=1pt, gray, opacity=0.25] (v7) -- (v15);

\node[vertex] (v16) [below left=25pt and 5pt of v8] {\tiny\color{black}16};
\node[vertex] (v17) [below right=25pt and 5pt of v8] {\tiny\color{black}17};
\node[vertex] (v18) [below left=25pt and 5pt of v9] {\tiny\color{black}18};
\node[vertex] (v19) [below right=25pt and 5pt of v9] {\tiny\color{black}19};
\node[vertex] (v20) [below left=25pt and 5pt of v10] {\tiny\color{black}20};
\node[vertex] (v21) [below right=25pt and 5pt of v10] {\tiny\color{black}21};
\node[vertex] (v22) [below left=25pt and 5pt of v11] {\tiny\color{black}22};
\node[vertex] (v23) [below right=25pt and 5pt of v11] {\tiny\color{black}23};
\node[vertex] (v24) [below left=25pt and 5pt of v12] {\tiny\color{black}24};
\node[vertex] (v25) [below right=25pt and 5pt of v12] {\tiny\color{black}25};
\node[vertex] (v26) [below left=25pt and 5pt of v13] {\tiny\color{black}26};
\node[vertex] (v27) [below right=25pt and 5pt of v13] {\tiny\color{black}27};
\node[vertex] (v28) [below left=25pt and 5pt of v14] {\tiny\color{black}28};
\node[vertex] (v29) [below right=25pt and 5pt of v14] {\tiny\color{black}29};
\node[vertex] (v30) [below left=25pt and 5pt of v15] {\tiny\color{black}30};
\node[vertex] (v31) [below right=25pt and 5pt of v15] {\tiny\color{black}31};

\draw[dashed, line width=1pt, gray, opacity=0.25] (v8) -- (v16);
\draw[line width=4pt, blue] (v8) -- (v17);
\draw[line width=1pt, red] (v8) -- (v17);
\draw[line width=4pt, blue] (v9) -- (v18);
\draw[line width=1pt, red] (v9) -- (v18);
\draw[dashed, line width=1pt, gray, opacity=0.25] (v9) -- (v19);
\draw[line width=4pt, blue] (v10) -- (v20);
\draw[line width=1pt, red] (v10) -- (v20);
\draw[dashed, line width=1pt, gray, opacity=0.25] (v10) -- (v21);
\draw[dashed, line width=1pt, gray, opacity=0.25] (v11) -- (v22);
\draw[line width=4pt, blue] (v11) -- (v23);
\draw[line width=1pt, red] (v11) -- (v23);
\draw[line width=4pt, blue] (v12) -- (v24);
\draw[line width=1pt, red] (v12) -- (v24);
\draw[dashed, line width=1pt, gray, opacity=0.25] (v12) -- (v25);
\draw[line width=2pt, blue] (v13) -- (v26);
\draw[dashed, line width=1pt, gray, opacity=0.25] (v13) -- (v27);
\draw[line width=2pt, blue] (v14) -- (v28);
\draw[dashed, line width=1pt, gray, opacity=0.25] (v14) -- (v29);
\draw[dashed, line width=1pt, gray, opacity=0.25] (v15) -- (v30);
\draw[line width=2pt, blue] (v15) -- (v31);

\end{tikzpicture}
		\caption{The same instance with full knowledge of left subtrees.}
		\label{fig:atpc-ins}
	\end{subfigure}
	\vskip1cm
	\begin{subfigure}{\textwidth}
		\centering

\begin{tikzpicture}

\tikzset{layer/.style={rectangle, rounded corners=5pt, draw, black, line width=1pt,  fill=black!10, inner sep=4pt}}
\tikzset{vertex/.style={circle, ForestGreen, fill=white, line width=2pt, draw, minimum width=10pt, minimum height=10pt, inner sep=1pt}}
\tikzset{choose/.style={rectangle, line width=1pt, rounded corners=2pt, draw, minimum width=50pt, minimum height=16pt, fill=black!10}}

\node[vertex] (v1) [] {\tiny\color{black}1};

\node[vertex] (v2) [below left=25pt and 110pt of v1] {\tiny\color{black}2};
\node[vertex] (v3) [below right=25pt and 110pt of v1] {\tiny\color{black}3};

\draw[dashed, line width=1pt, gray, opacity=0.25] (v1) -- (v2);
\draw[line width=2pt, red] (v1) -- (v3);

\node[vertex] (v4) [below left=25pt and 50pt of v2] {\tiny\color{black}4};
\node[vertex] (v5) [below right=25pt and 50pt of v2] {\tiny\color{black}5};
\node[vertex] (v6) [below left=25pt and 50pt of v3] {\tiny\color{black}6};
\node[vertex] (v7) [below right=25pt and 50pt of v3] {\tiny\color{black}7};

\draw[line width=4pt, blue] (v2) -- (v4);
\draw[line width=1pt, red] (v2) -- (v4);
\draw[dashed, line width=1pt, gray, opacity=0.25] (v2) -- (v5);
\draw[dashed, line width=1pt, gray, opacity=0.25] (v3) -- (v6);
\draw[line width=2pt, blue] (v3) -- (v7);

\node[vertex] (v8) [below left=25pt and 20pt of v4] {\tiny\color{black}8};
\node[vertex] (v9) [below right=25pt and 20pt of v4] {\tiny\color{black}9};
\node[vertex] (v10) [below left=25pt and 20pt of v5] {\tiny\color{black}10};
\node[vertex] (v11) [below right=25pt and 20pt of v5] {\tiny\color{black}11};
\node[vertex] (v12) [below left=25pt and 20pt of v6] {\tiny\color{black}12};
\node[vertex] (v13) [below right=25pt and 20pt of v6] {\tiny\color{black}13};
\node[vertex] (v14) [below left=25pt and 20pt of v7] {\tiny\color{black}14};
\node[vertex] (v15) [below right=25pt and 20pt of v7] {\tiny\color{black}15};

\draw[dashed, line width=1pt, gray, opacity=0.25] (v4) -- (v8);
\draw[line width=2pt, red] (v4) -- (v9);
\draw[dashed, line width=1pt, gray, opacity=0.25] (v5) -- (v10);
\draw[line width=2pt, red] (v5) -- (v11);
\draw[dashed, line width=1pt, gray, opacity=0.25] (v6) -- (v12);
\draw[line width=4pt, red] (v6) -- (v13);
\draw[line width=1pt, blue] (v6) -- (v13);
\draw[line width=2pt, red] (v7) -- (v14);
\draw[dashed, line width=1pt, gray, opacity=0.25] (v7) -- (v15);

\node[vertex] (v16) [below left=25pt and 5pt of v8] {\tiny\color{black}16};
\node[vertex] (v17) [below right=25pt and 5pt of v8] {\tiny\color{black}17};
\node[vertex] (v18) [below left=25pt and 5pt of v9] {\tiny\color{black}18};
\node[vertex] (v19) [below right=25pt and 5pt of v9] {\tiny\color{black}19};
\node[vertex] (v20) [below left=25pt and 5pt of v10] {\tiny\color{black}20};
\node[vertex] (v21) [below right=25pt and 5pt of v10] {\tiny\color{black}21};
\node[vertex] (v22) [below left=25pt and 5pt of v11] {\tiny\color{black}22};
\node[vertex] (v23) [below right=25pt and 5pt of v11] {\tiny\color{black}23};
\node[vertex] (v24) [below left=25pt and 5pt of v12] {\tiny\color{black}24};
\node[vertex] (v25) [below right=25pt and 5pt of v12] {\tiny\color{black}25};
\node[vertex] (v26) [below left=25pt and 5pt of v13] {\tiny\color{black}26};
\node[vertex] (v27) [below right=25pt and 5pt of v13] {\tiny\color{black}27};
\node[vertex] (v28) [below left=25pt and 5pt of v14] {\tiny\color{black}28};
\node[vertex] (v29) [below right=25pt and 5pt of v14] {\tiny\color{black}29};
\node[vertex] (v30) [below left=25pt and 5pt of v15] {\tiny\color{black}30};
\node[vertex] (v31) [below right=25pt and 5pt of v15] {\tiny\color{black}31};

\draw[dashed, line width=1pt, gray, opacity=0.25] (v8) -- (v16);
\draw[line width=4pt, blue] (v8) -- (v17);
\draw[line width=1pt, red] (v8) -- (v17);
\draw[line width=4pt, blue] (v9) -- (v18);
\draw[line width=1pt, red] (v9) -- (v18);
\draw[dashed, line width=1pt, gray, opacity=0.25] (v9) -- (v19);
\draw[line width=4pt, blue, dashed, opacity=0.5] (v10) -- (v20);
\draw[line width=1pt, red] (v10) -- (v20);
\draw[dashed, line width=1pt, gray, opacity=0.25] (v10) -- (v21);
\draw[dashed, line width=1pt, gray, opacity=0.25] (v11) -- (v22);
\draw[line width=4pt, blue, dashed, opacity=0.5] (v11) -- (v23);
\draw[line width=4pt, blue] (v12) -- (v24);
\draw[line width=1pt, red] (v12) -- (v24);
\draw[dashed, line width=1pt, gray, opacity=0.25] (v12) -- (v25);
\draw[line width=2pt, blue] (v13) -- (v26);
\draw[dashed, line width=1pt, gray, opacity=0.25] (v13) -- (v27);
\draw[line width=2pt, blue] (v14) -- (v28);
\draw[dashed, line width=1pt, gray, opacity=0.25] (v14) -- (v29);
\draw[dashed, line width=1pt, gray, opacity=0.25] (v15) -- (v30);
\draw[line width=2pt, blue] (v15) -- (v31);

\end{tikzpicture}
		\caption{The same instance with full knowledge of left subtrees and no knowledge of right subtrees.}
		\label{fig:atpc-del}
	\end{subfigure}
	\caption{An illustration of ATPC instances. Solid, blue edges are known to Alice and solid, red edges are known to Bob. Thick edges are owned in standard Tree Pointer Chasing while thin edges are known via augmentation. (For example, in \Cref{fig:atpc-del}, there are two \emph{overlapping} edges from node $6$ to node $13$. One is red and thick, meaning that Bob owns this edge in standard Tree Pointer Chasing, and the other is blue and thin, meaning that Alice knows this edge via augmentation.) Dashed, light-colored edges are forgotten during augmentation.}
	\label{fig:atpc}
\end{figure}

Note that Index can be viewed as Pointer Chasing on single-level trees.
To generalize it to multi-level trees, recall that in the standard Tree Pointer Chasing problem, one party owns all pointers in odd levels (that is, the first party gets as input an edge going out of each node in an odd level) and the other party owns all pointers in even levels. The parties' goal is to output the unique leaf node that can be reached using the parties' pointers. See~\Cref{fig:atpc-raw} for an example.

A natural attempt is to additionally give the owner of each pointer full knowledge of all the left subtrees, or equivalently all pointers owned by the other party in those subtrees. In other words, if a party has, as part of its input, the pointer connecting vertex $v$ to its $i$-th child, then the same party also gets all the pointers in the other party's input for the subtrees rooted at the first $i-1$ children of $v$. See~\Cref{fig:atpc-ins} for an illustration.
For example, in the illustration, since Bob has the pointer connecting the root to its second child, Bob also knows all Alice's pointers in the entire left subtree of the root.

\paragraph{Forgetting pointers.} We wish to prove a lower bound for the augmented Pointer Chasing problem on trees as described above. However, we next show that there is a subtle issue.
Suppose we want to prove the lower bound using the, by now standard, embedding arguments, showing that a protocol for instances with $r$ levels implies a protocol with one less message for instances with $r-1$ levels.
To do so, we sample an instance with $r$ levels as follows.
We denote by $(A_j,B_j)$ the subinstance corresponding to the $j$-th subtree (of the root) of the $r$-level instance we are sampling. We also denote by $(A',B')$ the input instance with $r-1$ levels that we attempt to solve. 
Alice and Bob publicly sample an index~$i$ and do the embedding by setting $(A_i,B_i) = (A',B')$.
$A_{< i}$ is also publicly sampled (note that this standard sampling respects our augmentation). 
To eliminate the first round of communication, Alice and Bob publicly sample Alice's first message $M_1$ (conditioned on $A_{< i}$).
In order to continue the simulation, the standard embedding argument would have the parties privately sample all remaining parts, namely $A_{> i}, B_{< i}, B_{> i}$.
Unfortunately, $A_{> i}$ and $B_{> i}$ may not be privately sampled (roughly following their original distributions) at the same time, due to possible high correlation.

To rectify the situation, we ``eliminate'' $B_{> i}$ by defining the {\em Augmented Tree Pointer Chasing} (ATPC) problem as follows: for each pointer, the party that owns it, also \begin{inparaenum}[(i)] \item knows \emph{everything} that the other party knows in subtrees to its left; and \item knows \emph{nothing} in subtrees to its right. \end{inparaenum} See~\Cref{fig:atpc-del} for an illustration.
For example, in the illustration, Alice ``forgets'' the pointer from node $10$ because it is in a subtree to the right of the pointer from node $2$.

We also emphasize that ``everything that the other party knows'' may not be equivalent to ``all pointers owned by the other party'', exactly because the other party may forget some of its originally owned pointers.
To see this, consider the pointers from nodes $10$ and $11$ in the illustration.
Before the augmentation, Alice knows both of them and Bob knows neither.
As we perform the augmentation bottom-up, Bob knows the one from node $10$ since it is in a subtree to the left of the pointer from node $5$.
Another level up, Alice forgets both of them due to the pointer from node $2$.
Note, however, that Bob still keeps his knowledge of the pointer from node $10$.
As a result, finally at the top level, Bob has the combined knowledge of both parties, including the pointer from node $10$, but not the one from node $11$.
In other words, Bob does not know the latter even though it is also in a subtree to the left of the pointer from the root.
Moreover, Bob's knowledge of the former is actually coming from himself in lower levels, but not from Alice.

A formal definition of ATPC is given in~\Cref{sec:few}.
Intuitively, the augmentation neither helps nor hurts the parties that attempt to solve an ATPC instance, as both parties should always follow the correct pointers. 
Indeed, we are able to prove an $r$ vs. $\Omega_r(n^{1/r})$ round-communication tradeoff for trees with $n$ leaf nodes, using standard information-theoretic tools.

\paragraph{Reducing ATPC to $\mur$ and the role of augmentation.} 

Next, we wish to show a lower bound for $\mur$ by proving that $\mur$ is even harder than ATPC. The reduction is as follows.
Given an ATPC instance, Alice is constructing the larger set $A$ (corresponding to insertions for MST) and Bob is constructing the smaller set $B$ (corresponding to deletions for MST).
The universe contains all the leaf nodes of the ATPC instance, sequentially ordered from left to right, and the goal is to have $\min(A \setminus B)$ being the leaf node induced by the pointers in the ATPC problem.
Imagine the parties perform the construction of the sets $A$ and $B$ ``bottom-up'' in the following sense.
Suppose the current pointer $i$ is known to Alice (a similar argument applies to the case in which Bob knows the current pointer).
Also, assume that the parties have already constructed $A_1,\ldots,A_w$ and $B_1,\ldots,B_w$ where $w$ is the arity of the ATPC tree and $(A_j,B_j)$ is the $\mur$ instance constructed for the $j$-th subtree of the ATPC tree, and they want to combine all these sets to obtain $A,B$. 

Now, we may wish for Bob to set $B = B_1 \cup \cdots \cup B_i$ and for Alice to set $A = B_1 \cup \cdots \cup B_{i-1} \cup A_i$, as this would imply $A \setminus B = A_i \setminus B_i$, while the promise that $B \subseteq A$ in the definition of $\mur$ is satisfied. 
Note that Alice can indeed compute this set $A$ {\em thanks to the augmentation} that gives her $B_1, \cdots, B_{i-1}$. In fact, this is the exact reason for the augmentation.
Unfortunately, though, Bob cannot compute $B$ as he does not know $i$.
Nevertheless, it can be easily remedied by setting $A = B_1 \cup \cdots \cup B_{i-1} \cup A_i \cup A'$ and $B = B_1 \cup \cdots \cup B_w$, where $A'$ is the set of all leaf nodes in subtrees $i+1,\ldots,w$. 

\paragraph{Weights.}
Since the number of weights in our MST instances is essentially the number of leaf nodes in the ATPC instances, our MST construction only uses polynomially many integer weights.
We note that this is necessary due to the result of~\cite{AhnGM12a}, as otherwise there is a single pass streaming algorithm that finds an MST in $n^{1+o(1)}$ space.
Specifically, an MST can be incrementally found by considering all edges of weight $i$ and applying the Spanning Forest algorithm of~\cite{AhnGM12a} at the $i$-th step.
This can be implemented in a single pass by maintaining $W$ independent copies of the sketch used for the Spanning Forest algorithm, resulting in an $\tilde{O}(nW)$-space algorithm.

\paragraph{Computing the MST weight with large edge weights.} So far, we are able to lower bound the search version of MST.
We note that the construction shown in~\Cref{fig:search} can be readily adapted for computing the weight of MSTs if \emph{exponential} edge weights were allowed\footnote{This is not an issue for the search version of MST as even linear edge weights are sufficient to ensure a unique MST, up to edges in the clique.}: edges incident on the $j$-th non-clique vertex have weights in the order of $n^j$, so that the minimum weight edge that is not deleted, for each non-clique vertex, can be uniquely recovered from the MST weight alone.
However, exponential edge weights would lead to a polynomial overhead in space requirement, which is unaffordable for streaming algorithms.
So, we explore the decision version of MST in the following, while keeping the edge weights polynomial.

\subsection{The Decision Version}\label{sec:decision}

\paragraph{Decisional $\mur$.} We next proceed to outline our lower bound for the algorithmically-easier decision version of the MST problem. 
Since there exist efficient algorithms, even with a single pass, for \emph{approximating} the weight of MSTs (e.g.~\cite{AhnGM12a}), we should expect hard instances for the decision version to have MST weights concentrated within a small range.
So the following attempt seems plausible.
Let $e_j$ be the minimum edge weight for the $j$-th non-clique vertex, and $z_j$ the parity of $e_j$.
Also let $T = \sum_j e_j - \sum_j z_j$.
Then the weight of MSTs is always between $T$ and $T+k$, where $k=n/2$ is the number of non-clique vertices.
In the above, we have argued that finding $e_j$ is hard for a fixed $j$.
With little additional effort we can show that computing $z_j$ is also hard.\footnote{Recall that the lower bound on $\mur$ is derived via ATPC. Roughly speaking, we may view the bottom level as a composition of two sublevels, one of which is binary.}
We denote by $\dmur$ the corresponding decisional universal relation problem, where one needs to compute the parity of the minimum element in $A \setminus B$.
We remark that this attempt is in line with~\cite{Yu21}, in which a decision version of Universal Relation, $\dur$, is utilized to obtain optimal lower bounds for Connectivity.

\paragraph{A majority lemma?} One may hope that our final result would again follow from a direct-sum (or, more accurately, ``{\em majority lemma}'') type argument: hardness of some boolean function $f$ implies hardness of computing the {\em majority} of $k$ copies of $f$\footnote{For simplicity, we may assume throughout this section that $f$ is ``balanced'' in the sense that it evaluates to $0$ on exactly half of possible inputs and to $1$ on the other half.}.
This is because given such a majority lemma, we can simply set the threshold to be $T+k/2$.
It is easy to see that the weight of MSTs exceeds $T+k/2$ if and only if the majority of the $k$ parity bits $z_j$ is $1$.

\paragraph{Fixing the threshold at the price of correlating the $\dmur$ instances.} To our disappointment, this approach has major problems.
One notable issue is that $T$ is ``instance dependent'', and is not a predetermined value, and therefore the threshold $T+k/2$ is also instance dependent.
This is indeed a problem as, in the reduction in \Cref{fig:search}, the parties would not know the threshold value required for the streaming MST instance.
In other words, we don't even have a well-defined input for the decision version of MST! 
To circumvent this, we add one special edge of weight $T'=C-T$ to the graph, where $C$ is a sufficiently large number to ensure $T'$ is positive. 
This way, we are always comparing the weight of the MSTs with a {\em fixed} number $C+k/2$.

In the communication setting, this addition is equivalent to {\em revealing $T$} to both parities (implemented as an extra part of input), which {\em correlates all $k$ copies of $\dmur$}.
Since a direct-sum style argument typically deals with independent copies, we now need to ``get rid'' of $T$.
Note that $T=\poly(n)=\poly(k)$.
This renders it impossible to brute force over all possible values of $T$ due to the communication constraint.

Another way of getting rid of $T$ would be to make a random guess at $T$ and output randomly if the guess is wrong (with very small communication overhead for verifying the guess). However, this approach has the following major shortcoming:
the random guessing reduces the advantage (over $1/2$) by a factor of $T=\poly(k)$ but {\em a majority lemma can never hold in such a low advantage regime!}
Specifically, the following may not be true:
\begin{quote}
	\textit{If computing $f$ with success probability $3/4$ requires $C$ communication, then computing the majority of $k$ copies of $f$ with success probability $1/2+1/k$ requires $\tilde{\Omega}(kC)$ communication.}
\end{quote}
What's even worse, is that $C$ communication is sufficient to achieve success probability $1/2+\Theta(1/\sqrt{k})$.
To see this, suppose $f$ evaluates to $0$ on exactly half of the first $k-1$ copies and $1$ on the other half, and then the majority is solely determined by the output of the last copy.
Now consider the protocol that simply computes the value of the last copy and outputs it as the majority.
It succeeds whenever the single copy protocol succeeds and thus has constant advantage ($3/4-1/2=1/4$ to be exact) in the above case, which occurs with probability $\Theta(1/\sqrt{k})$ due to properties of binomial distributions, and is equivalent to a random guess in all other cases as the majority is already determined by the first $k-1$ copies (recall that we assume $f$ to be balanced). 
So we cannot hope for a majority lemma that works with advantage well below $\Theta(1/\sqrt{k})$.
This dooms our attempt as we are requiring even much lower advantage. 

\paragraph{Majority Lemma with hint via XOR Lemma with hint.}

We work around the above limitation by a different approach.
Instead of directly getting rid of $T$ and seeking a majority lemma with low advantage (which turns out to be nonexistent), we convert majority computation into XOR computation by a simple process (with $T$ revealed).
Only after that, we again guess $T$ and then utilize an XOR lemma with low advantage which indeed exists.
As will be seen later, this alternative approach can be viewed as a majority lemma with high advantage (close to $1/2$).


To prove this latter majority lemma, we start from the beautiful recent work~\cite{Yu22} that provides a strong XOR lemma in which advantage decreases exponentially in $k$. 
We then consider the following process for computing XOR from majority.
If the number of $1$'s is at most $k/2$ (so the majority is $0$), return the parity of $k/2$ (assume that $k$ is even), and otherwise return the parity of $k/2+1$.
Intuitively, the probability of having exactly $i$ $1$'s is slightly larger than having $i-1$, for $i \le k/2$.
So this process should have certain advantage over $1/2$.
Indeed, again by properties of binomial distributions, this advantage can be shown to be $\Theta(1/\sqrt{k})$, assuming that the computation of majority is perfect.
In general, we can prove that a protocol for computing majority with success probability $1-\epsilon$ implies a protocol for computing XOR with success probability $1/2-\epsilon+\Theta(1/\sqrt{k})$. 
Since the XOR lemma of \cite{Yu22} proves that a protocol for computing the XOR with success probability $1/2-\epsilon+\Theta(1/\sqrt{k})$ (or even $1/2+\exp(-k)$) is costly, it also implies that the computation of the majority with success probability $1-\epsilon$ is costly. Our entire proof now works as follows.
 
\begin{enumerate}
	\item Prove a lower bound on $\dmur$.
	\item Apply the XOR lemma of~\cite{Yu22} to show it is also hard to compute the XOR of $k$ copies of $\dmur$, with success probability $1/2+1/\poly(k)$.
		This hardness continues to hold with $T$ revealed, which we call a ``hint'' in our proof.
	\item Using the above process, we get a lower bound for computing the majority of $k$ copies of $\dmur$, with success probability $1-1/\poly(k)$,  and also with hint $T$.
	\item Finally, a streaming lower bound for the decision version of MST is derived by our reduction (up to logarithmic factors resulted from boosting the success probability).
\end{enumerate}
All the above ideas are formalized in~\Cref{sec:few}.
At a high level, what we really use, is roughly a majority lemma of the following form, which has a very weak probability guarantee that is enough for us:
\begin{quote}
	\textit{If computing $f$ with success probability $3/4$ requires $C$ communication, then computing the majority of $k$ copies of $f$ with success probability $1-1/\poly(k)$ requires $\tilde{\Omega}(kC)$ communication.}
\end{quote}
We note, however, that we need such a lemma that also works when $T$ is revealed. As we claimed before, to prove a majority lemma that works when $T$ is revealed, we can guess $T$, but then need to prove a majority lemma with a very small advantage. Likewise, to show an XOR lemma that works when $T$ is revealed, we can guess $T$ and prove an XOR lemma for very small advantages. Luckily, unlike the case for majority, such an XOR lemma can be proved. Indeed, the XOR lemma of \cite{Yu22} has a strong enough probability guarantee.  


Unfortunately, all the above has not yet led to an optimal pass lower bound.
Specifically, the XOR lemma of~\cite{Yu22} shows that computing the XOR of $k$ copies requires roughly $k/r^{O(r)}$ times the communication for computing a single copy, where $r$ is the number of communication rounds.
This loss of an $r^{O(r)}$ factor is problematic as the final lower bound that can be obtained is roughly $n^{1+1/r}/r^{O(r)}$, which only works for $r$ up to $\sqrt{\log n / \log\log n}$.
For comparison, \cite{AhnGM12a} presents a semi-streaming algorithm of $O(\log n / \log\log n)$ passes.
So, there is still a gap between the lower and the upper bounds.
Furthermore, the loss in communication turns out to be the sole barrier for closing this gap, in the sense that we can prove a tight lower bound if the $r^{O(r)}$ factor could be reduced to $\poly(r)$.
We address this challenge in the rest of this section, by proving a {\em multi-party} XOR lemma (rather than a two-party one) with a better dependence on the number of rounds. 

\subsection{Multi-Party XOR Lemma}\label{sec:multi-xor}

\paragraph{The XOR lemma we need.} As indicated above, an ideal XOR lemma (in the standard two-party setting) that is sufficient for our purpose is of the following form: computing the XOR of $k$ copies requires $k/\poly(r)$ times the communication for computing a single copy, to achieve $1/\poly(k)$ advantage.
Note that such an XOR lemma does not necessarily improve upon~\cite{Yu22} as it only requires a polynomial advantage decay.
Nevertheless, to the best of our knowledge, the existence of such an XOR lemma is still unknown.

We prove such a lemma in the multi-party setting. We note that we opt not to restrict ourselves in the two-party setting as our ultimate goal is to prove streaming lower bounds and multi-party settings are usually easier to work with.
Nevertheless, our multi-party XOR lemma may be of independent interest as well since it works entirely in the communication setting, with no reference to streaming.

\paragraph{Separating amplification of communication and of advantage.}
To get our multi-party XOR lemma, we decompose it into two \emph{independent} parts: {\em amplification of communication} and {\em amplification of advantage}.
More specifically, up to $\poly(r)$ factors, amplification of communication means:
\begin{quote}
	\textit{If computing $f$ with success probability $1/2+\delta$ requires $C$ communication, then computing the XOR of $k_1$ copies of $f$ with success probability $1/2+\delta+\epsilon$ requires $\tilde{\Omega}_\epsilon(k_1C)$ communication,}
\end{quote}
and amplification of advantage means:
\begin{quote}
	\textit{If computing $f$ with success probability $3/4$ requires $C$ communication, then computing the XOR of $k_2$ copies of $f$ with success probability $1/2+\exp(-\Omega(k_2))$ requires $\tilde{\Omega}(C)$ communication.}
\end{quote}
Intuitively, this decomposition is possible because the XOR of many XOR computations is equivalent to a single XOR computation.
Furthermore, our desired XOR lemma, up to logarithmic factors, follows from combining amplification of communication with $k_1 = \Theta(k/\log k)$ and amplification of advantage with $k_2 = \Theta(\log k)$.

\paragraph{Amplification of communication.} As to amplification of communication, it can be accomplished using known (round-preserving) compression schemes (e.g., \cite{JainPY12,BravermanRWY13}) in the standard two-party setting.
However, we do emphasize that known compression schemes all seem to have a linear (or even polynomial) dependence on $1/\epsilon$ in the communication if we want an $\epsilon$-simulation.
This essentially means amplification of communication has to be performed before amplification of advantage.
Otherwise, communication would suffer a polynomial blowup in order to preserve the already amplified advantage.

\paragraph{Amplification of advantage.} For amplification of advantage, inspiration is drawn from the streaming XOR lemma by~\cite{AssadiN21}.
Their result shows that computing the XOR of $k$ copies in the streaming setting with the same space constraint as for a single copy, can only achieve advantage exponentially small in $k$.
Moreover, streaming algorithms are viewed as multi-party communication protocols in their proof.
This enables us to adapt their techniques to prove a multi-party communication version: computing the XOR of $k$ copies with ($2k$ parties and) the same total communication as for a single copy (with two parties), can only achieve advantage exponentially small in $k$.
Combined with amplification of communication, it finally yields a multi-party XOR lemma with the desired parameters.


We also remark that the streaming XOR lemma of~\cite{AssadiN21} applies to streams in which $k$ copies arrive \emph{sequentially}, i.e., one complete stream followed by another.
For our MST construction, this means insertions of the first non-clique vertex is followed by deletions of the same vertex, and then insertions and deletions of the second non-clique vertex and so on.
In contrast, our version for multi-party communication has an ``interleaved'' input order in the sense that part of the first copy (insertions for the first non-clique vertex) is followed by part of the second copy (insertions for the second non-clique vertex) and so on for all other copies, and the remaining part of the first copy (deletions for the first non-clique vertex) only comes after that.
Put it another way, all the Alices communicate before all the Bobs.
Consequently, the streams resulted from our proof have the {\em simplest form}: all insertions arrive before all deletions.




\section{A Lower Bound in Few Passes}\label{sec:few}

As a warmup, we first prove the following weaker version of~\Cref{res:main} for only few passes.
It already contains many of the critical ideas for fully proving~\Cref{res:main}, while also identifying the key barrier in getting a proof for even more passes.

\begin{theorem}[Weaker version of~\Cref{res:main}]\label{thm:lb-few}
	For $p = o(\sqrt{\frac{\log n}{\log\log n}})$, any $p$-pass dynamic streaming algorithm for solving $\mst_n$ with probability $2/3$ requires $\Omega(\frac{n^{1+\frac{1}{2p-1}}}{p^{O(p)}\log n})$ space.
\end{theorem}

\noindent
We remark that the upper bound on edge weights in~\Cref{res:main} will be seen in the proof of~\Cref{clm:hint-atpc-maj-mst}.

\subsection{Augmented Tree Pointer Chasing}

The proof of~\Cref{thm:lb-few} is via a communication problem named Augmented Tree Pointer Chasing.

\begin{definition}
	For $d,w \ge 1$, the two-party problem $\atpc_{d,w}$ is defined recursively as follows.
	\begin{enumerate}
		\item For $d=1$, Alice is given as input $A^{(1)} \in \set{0,1}^w$ and Bob is given as input $B^{(1)} = (i^{(1)},A^{(1)}_{< i^{(1)}})$, where $i^{(1)} \in [w]$. They are required to output $A^{(1)}_{i^{(1)}}$.
		\item For $d>1$, Alice is given as input $A^{(d)} = b^{(d-1)}_{\le w}$ and Bob is given as input $B^{(d)} = (i^{(d)},a^{(d-1)}_{\le i^{(d)}},b^{(d-1)}_{< i^{(d)}})$, where $i^{(d)} \in [w]$ and $(a^{(d-1)}_j,b^{(d-1)}_j)$ for $j \in [w]$ is an instance of $\atpc_{d-1,w}$\footnote{For $j > i^{(d)}$, $a^{(d-1)}_j$ is imaginary and given to neither party.}. They are required to output the answer to $(a^{(d-1)}_{i^{(d)}},b^{(d-1)}_{i^{(d)}})$ as an instance of $\atpc_{d-1,w}$.
	\end{enumerate}
	For $k \ge 1$, $\atpc^{\oplus k}_{d,w}$ denotes the $k$-fold XOR version of $\atpc_{d,w}$, and similarly $\atpc^{\# k}_{d,w}$ denotes the $k$-fold majority version of $\atpc_{d,w}$.

	Each $\atpc_{d,w}$ instance can be naturally visualized on a depth-$d$, $w$-ary tree with $i$'s being pointers of corresponding levels.
	Suppose the leaf nodes are numbered from $1$ to $w^d$.
	Starting from the root and following the pointers will lead to a unique leaf node $t \in [w^d]$, which is called the \emph{target} of this instance.
	For either of the $k$-fold versions of $\atpc_{d,w}$, both parties may additionally be given the \emph{hint} $T = \sum_{j \in [k]} t_{j}$ as part of their input, where $t_{j}$ is the target of the $j$-th $\atpc_{d,w}$ instance.
	The resulting problems are denoted by $\hintatpc^{\oplus k}_{d,w}$ and $\hintatpc^{\# k}_{d,w}$, respectively.
\end{definition}

At a high level, in an instance of $\atpc_{d,w}$, Alice owns all pointers at even levels while Bob owns all pointers at odd levels.
It only differs from the stardard Tree Pointer Chasing problem by performing the following modification to each internal node: the owner of a pointer is additionally given the other party's knowledge of subtrees to the left of the pointer, while losing any knowledge of subtrees to the right of the pointer.
The modification is performed bottom-up.
In other words, the effect of ancestors supersedes that of descendants.
Intuitively, the extra information about subtrees to the left cannot help while the lost information about subtrees to the right cannot hurt, as the owner of the current node should always follow the pointer.
Also note that $\atpc_{1,w}$ is exactly the same as the well-studied Augmented Index problem.
For $d > 1$, $\atpc_{d,w}$ can be naturally viewed as its multi-round generalization.

The hard input distributions and corresponding lower bounds are as follows.

\begin{Distribution}\label{dist:atpc}
	For $d,w \ge 1$, the hard input distribution $\cD_{d,w}$ is defined recursively as follows.
	\begin{enumerate}
		\item For $d=1$, Alice is given as input $A^{(1)}$ and Bob is given as input $B^{(1)} = (i^{(1)},A^{(1)}_{< i^{(1)}})$, where $i^{(1)}$ is sampled from $[w]$ uniformly at random and $A^{(1)}$ is independently sampled from $\set{0,1}^w$ uniformly at random.
		\item For $d>1$, Alice is given as input $A^{(d)} = b^{(d-1)}_{\le w}$ and Bob is given as input $B^{(d)} = (i^{(d)},a^{(d-1)}_{\le i^{(d)}},b^{(d-1)}_{< i^{(d)}})$, where $i^{(d)}$ is sampled from $[w]$ uniformly at random and $(a^{(d-1)}_j,b^{(d-1)}_j)$ for $j \in [w]$ is independently sampled from $\cD_{d-1,w}$.
	\end{enumerate}
\end{Distribution}

\begin{lemma}\label{lem:lb-atpc}
	For $d,w \ge 1$ and $\epsilon \in [0,1/2]$, it holds that
	\[
		\bD^{(d)}_{\cD_{d,w},\frac{1}{2}+\epsilon}(\atpc_{d,w}) \ge \frac{\epsilon^2 w}{d}.
	\]
\end{lemma}

\begin{Distribution}\label{dist:atpc-k}
	For $k,d,w \ge 1$, the hard input distribution $\cD^k_{d,w}$ is defined as follows.
	Alice is given as input $A = a^{(d)}_{\le k}$ and Bob is given as input $B = b^{(d)}_{\le k}$, where $(a^{(d)}_j,b^{(d)}_j)$ for $j \in [k]$ is independently sampled from $\cD_{d,w}$.
\end{Distribution}

\begin{lemma}\label{lem:lb-hint-atpc-maj}
	For $w \ge 1$, $d = o(\log w / \log\log w)$, and $k = \omega(d\log w)$, it holds that
	\[
		\bD^{(d)}_{\cD^k_{d,w},\frac{2}{3}}(\hintatpc^{\# k}_{d,w}) = \Omega\paren{\frac{kw}{d^{O(d)}\log k}}.
	\]
\end{lemma}

We first prove~\Cref{thm:lb-few} in~\Cref{sec:proof-lb-few}, assuming the lower bounds for Augmented Tree Pointer Chasing.
Proofs of the above lower bounds are shown in \Cref{sec:proof-lb-atpc,sec:proof-lb-hint-atpc-maj}, respectively.

\subsection{Proof of~\Cref{thm:lb-few}}\label{sec:proof-lb-few}

In this section, we present a proof of~\Cref{thm:lb-few} via the following claim.

\begin{claim}\label{clm:hint-atpc-maj-mst}
	For $k,p,w,S \ge 1$ and $\epsilon \in [0,1]$, if there exists a $p$-pass, $S$-space dynamic streaming algorithm for solving $\mst_{k+w^{2p-1}+1}$ with probability $\epsilon$, then there also exists a $(2p-1)$-round, $(2p-1)S$-communication protocol for solving $\hintatpc^{\# k}_{2p-1,w}$ with probability $\epsilon$ over $\cD^k_{2p-1,w}$.
\end{claim}

Before proving~\Cref{clm:hint-atpc-maj-mst}, we show that it indeed implies~\Cref{thm:lb-few}.

\begin{proof}[Proof of~\Cref{thm:lb-few}]
	Fix a $p$-pass dynamic streaming algorithm for solving $\mst_n$ with probability $2/3$ that has space $S$.
	Let $k = (n-1)/2$, $d = 2p-1$, $w = (n-k-1)^{1/d}$, and $C = dS$.
	Applying the reduction of~\Cref{clm:hint-atpc-maj-mst}, we get a $d$-round protocol for solving $\hintatpc^{\# k}_{d,w}$ with probability $2/3$ over $\cD^k_{d,w}$ that has communication $C$.
	On the other hand, \Cref{lem:lb-hint-atpc-maj} implies
	\[
		C = \Omega\paren{\frac{kw}{d^{O(d)}\log k}},
	\]
	or equivalently,
	\[
		S = \Omega\paren{\frac{n^{1+\frac{1}{2p-1}}}{p^{O(p)}\log n}},
	\]
	as claimed.
\end{proof}

We remark that the assumption $p = o(\sqrt{\log n / \log\log n})$ of~\Cref{thm:lb-few} is nessesary in the above proof for satisfying the condition $d = o(\log w / \log\log w)$ of~\Cref{lem:lb-hint-atpc-maj}.
Moreover, a closer look at the proofs in~\Cref{sec:proof-lb-atpc,sec:proof-lb-hint-atpc-maj}  will reveal that this constraint comes solely from the $r^{O(r)}$-fold decrease in communication when using the XOR lemma of~\cite{Yu22}.
If the loss factor were reduced to $\poly(r)$ (meaning a better XOR lemma), the constraint would then be relaxed to $d = w^{o(1)}$.
In turn, this would be sufficient for proving a lower bound for up to $p = o(\log n / \log\log n)$ passes (and thus the full verion of our main result).
Nevertheless, as will be seen in~\Cref{sec:optimal}, we actually take a two-step approach in the absence of such an ideal XOR lemma.
At a high level, we will perform the amplification of communication and error probability separately.

The rest of this section constitutes a proof of~\Cref{clm:hint-atpc-maj-mst}.
Let $d=2p-1$, $C=dS$, and $n=k+w^d+1$.
Fix a dynamic streaming algorithm $\pi$ as described in the claim.
In the following, we construct a protocol $\tau$ for solving $\hintatpc^{\# k}_{d,w}$ with the desired properties.
On input $((A=a^{(d)}_{\le k},T),(B=b^{(d)}_{\le k},T))$, $\tau$ simulates $\pi$ on the following dynamic stream, where $\sender$ and $\receiver$ are defined in~\Cref{alg:sender} and~\Cref{alg:receiver}, respectively ($p=0$ represents Alice and $p=1$ represents Bob); see~\Cref{fig:sender-receiver} for an illustration of the functions and~\Cref{fig:mst} for an illustration of the reduction.
The threshold given to $\pi$ is to be determined.
\begin{enumerate}
	\item Insert an edge $(1,n)$ with weight $2kw^d-2T+1$.
	\item For $u < v \in [w^d]$, insert an edge $(k+u,k+v)$ with weight $1$.
	\item For $j \in [k]$ and $t \in \sender(d,w,a^{(d)}_j,0)$, insert an edge $(j,k+\ceil{t/2})$ with weight $t+1$.
	\item For $j \in [k]$ and $t \in \receiver(d,w,b^{(d)}_j,1)$, delete the edge $(j,k+\ceil{t/2})$ with weight $t+1$.
\end{enumerate}

\begin{Algorithm}\label{alg:sender}
	The function $\sender(d,w,A^{(d)},p)$.
	\begin{itemize}
		\item $d=1$: We have $A^{(1)} \in \set{0,1}^w$.
			\begin{itemize}
				\item $p=0$: Return
					\[
						\set{2j-A^{(1)}_j \mid j \in [w]}.
					\]
				\item $p=1$: Return
					\[
						[2w] \setminus \set{2j-A^{(1)}_j \mid j \in [w]}.
					\]
			\end{itemize}
		\item $d>1$: We have $A^{(d)} = b^{(d-1)}_{\le w}$, where $b^{(d-1)}_j$ for $j \in [w]$ is a valid input to Bob for $\atpc_{d-1,w}$. Return
			\[
				\bigcup_{j \in [w]} \set{2(j-1) \cdot w^{d-1} + t \mid t \in \receiver(d-1,w,b^{(d-1)}_j,p)}.
			\]
	\end{itemize}
\end{Algorithm}

\begin{Algorithm}\label{alg:receiver}
	The function $\receiver(d,w,B^{(d)},p)$.
	\begin{itemize}
		\item $d=1$: We have $B^{(1)} = (i^{(1)},A^{(1)}_{< i^{(1)}})$, where $i^{(1)} \in [w]$ and $A^{(1)}_j \in \set{0,1}$ for $j \in [i^{(1)}-1]$.
			\begin{itemize}
				\item $p=0$: Return
					\[
						[2w] \setminus \set{2j-A^{(1)}_j \mid j \in [i^{(1)}-1]}.
					\]
				\item $p=1$: Return
					\[
						\set{2j-A^{(1)}_j \mid j \in [i^{(1)}-1]}.
					\]
			\end{itemize}
		\item $d>1$: We have $B^{(d)} = (i^{(d)},a^{(d-1)}_{\le i^{(d)}},b^{(d-1)}_{< i^{(d)}})$, where $i^{(d)} \in [w]$, $(a^{(d-1)}_j,b^{(d-1)}_j)$ for $j \in [i^{(d-1)}-1]$ is a valid instance of $\atpc_{d-1,w}$, and $a^{(d-1)}_{i^{(d)}}$ is a valid input to Alice for $\atpc_{d-1,w}$. If $p=0$, return
			\begin{align*}
				& \paren{\bigcup_{j \in [i^{(d)}-1]} \set{2(j-1) \cdot w^{d-1} + t \mid t \in \receiver(d-1,w,b^{(d-1)}_j,1)}}\\
				& \hspace{1cm} \cup \set{2(i^{(d)}-1) \cdot w^{d-1} + t \mid t \in \sender(d-1,w,a^{(d-1)}_{i^{(d)}},0)}\\
				& \hspace{1cm} \cup [2i^{(d)} \cdot w^{d-1} + 1, 2w^d],
			\end{align*}
			and if $p=1$, return
			\begin{align*}
				& \paren{\bigcup_{j \in [i^{(d)}-1]} \set{2(j-1) \cdot w^{d-1} + t \mid t \in \receiver(d-1,w,b^{(d-1)}_j,0)}}\\
				& \hspace{1cm} \cup \set{2(i^{(d)}-1) \cdot w^{d-1} + t \mid t \in \sender(d-1,w,a^{(d-1)}_{i^{(d)}},1)}.
			\end{align*}
	\end{itemize}
\end{Algorithm}

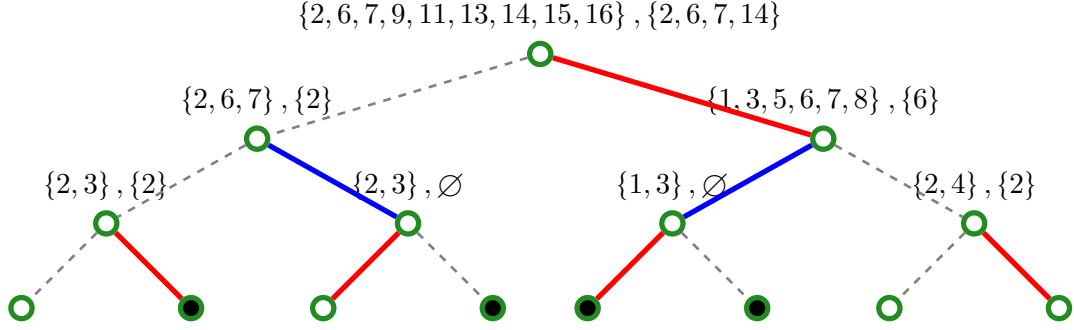
\begin{figure}[!tbh]
	\centering

\begin{tikzpicture}

\tikzset{layer/.style={rectangle, rounded corners=5pt, draw, black, line width=1pt,  fill=black!10, inner sep=4pt}}
\tikzset{vertex/.style={circle, ForestGreen, fill=white, line width=2pt, draw, minimum width=8pt, minimum height=8pt, inner sep=0pt}}
\tikzset{choose/.style={rectangle, line width=1pt, rounded corners = 2pt, draw, minimum width=40pt, minimum height=16pt, fill=black!10}}

\node[vertex, label={$\set{2,6,7,9,11,13,14,15,16},\set{2,6,7,14}$}] (v1) [] {};

\node[vertex, label={$\set{2,6,7},\set{2}$}] (v2) [below left=25pt and 100pt of v1] {};
\node[vertex, label={$\set{1,3,5,6,7,8},\set{6}$}] (v3) [below right=25pt and 100pt of v1] {};

\draw[dashed, line width=1pt, gray, opacity=0.25] (v1) -- (v2);
\draw[line width=2pt, red] (v1) -- (v3);

\node[vertex, label={$\set{2,3},\set{2}$}] (v4) [below left=25pt and 50pt of v2] {};
\node[vertex, label={$\set{2,3},\emptyset$}] (v5) [below right=25pt and 50pt of v2] {};
\node[vertex, label={$\set{1,3},\emptyset$}] (v6) [below left=25pt and 50pt of v3] {};
\node[vertex, label={$\set{2,4},\set{2}$}] (v7) [below right=25pt and 50pt of v3] {};

\draw[dashed, line width=1pt, gray, opacity=0.25] (v2) -- (v4);
\draw[line width=2pt, blue] (v2) -- (v5);
\draw[line width=2pt, blue] (v3) -- (v6);
\draw[dashed, line width=1pt, gray, opacity=0.25] (v3) -- (v7);

\node[vertex, fill=white] (v8) [below left=25pt and 25pt of v4] {};
\node[vertex, fill=black] (v9) [below right=25pt and 25pt of v4] {};
\node[vertex, fill=white] (v10) [below left=25pt and 25pt of v5] {};
\node[vertex, fill=black] (v11) [below right=25pt and 25pt of v5] {};
\node[vertex, fill=black] (v12) [below left=25pt and 25pt of v6] {};
\node[vertex, fill=black] (v13) [below right=25pt and 25pt of v6] {};
\node[vertex, fill=white] (v14) [below left=25pt and 25pt of v7] {};
\node[vertex, fill=white] (v15) [below right=25pt and 25pt of v7] {};

\draw[dashed, line width=1pt, gray, opacity=0.25] (v4) -- (v8);
\draw[line width=2pt, red] (v4) -- (v9);
\draw[line width=2pt, red] (v5) -- (v10);
\draw[dashed, line width=1pt, gray, opacity=0.25] (v5) -- (v11);
\draw[line width=2pt, red] (v6) -- (v12);
\draw[dashed, line width=1pt, gray, opacity=0.25] (v6) -- (v13);
\draw[dashed, line width=1pt, gray, opacity=0.25] (v7) -- (v14);
\draw[line width=2pt, red] (v7) -- (v15);

\end{tikzpicture}
	\caption{An illustration of functions $\sender$ and $\receiver$, for $d=3$ and $w=2$. Blue edges are the pointers owned by Alice (not via augmentation) while red edges are the pointers owned by Bob (not via augmentation). Unfilled leaf nodes have value $0$ while filled leaf nodes have value $1$. Each internal node is labeled by the set of insertions (for Alice), followed by the set of deletions (for Bob), with respect to the subinstance represented by its subtree, where the owner of the pointer computes $\receiver$ and the other party computes $\sender$.}
	\label{fig:sender-receiver}
\end{figure}

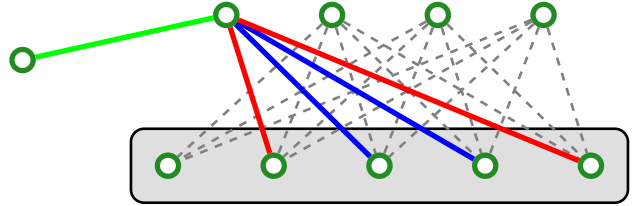
\begin{figure}[!tbh]
	\centering

\begin{tikzpicture}

\tikzset{layer/.style={rectangle, rounded corners=5pt, draw, black, line width=1pt,  fill=black!10, inner sep=4pt}}
\tikzset{vertex/.style={circle, ForestGreen, fill=white, line width=2pt, draw, minimum width=8pt, minimum height=8pt, inner sep=0pt}}
\tikzset{choose/.style={rectangle, line width=1pt, rounded corners = 2pt, draw, minimum width=40pt, minimum height=16pt, fill=black!10}}

\begin{scope}[local bounding box=bb]
	\node[vertex] (w1) {};
	\foreach \i in {2,...,5}
	{
		\pgfmathtruncatemacro{\ip}{\i-1}
		\node[vertex] (w\i) [right=30pt of w\ip] {};
	}
\end{scope}

\begin{scope}[on background layer]
	\draw[layer, fill=gray!25] ($(bb.south west) - (10pt,10pt)$) rectangle ($(bb.north east) + (10pt,10pt)$);
\end{scope}

\node[vertex] (j1) [above right=50pt and 15pt of w1] {};

\foreach \i in {2,...,4}
{
	\pgfmathtruncatemacro{\ip}{\i-1}
	\node[vertex] (j\i) [right=30pt of j\ip] {};
}

\node[vertex] (n) [below left=10pt and 70pt of j1] {};

\foreach \j in {2,...,4}
{
	\foreach \k in {1,...,5}
	{
		\draw[dashed, line width=1pt, gray, opacity=0.25] (j\j) -- (w\k);
	}
}

\draw[line width=2pt, green] (n) -- (j1);

\draw[line width=2pt, red] (j1) -- (w2);
\draw[line width=2pt, blue] (j1) -- (w3);
\draw[line width=2pt, blue] (j1) -- (w4);
\draw[line width=2pt, red] (j1) -- (w5);

\end{tikzpicture}
	\caption{An illustration of the reduction from $\hintatpc^{\# k}_{d,w}$ to $\mst_n$, for $k=4$, $w^d=5$, and $n=10$. Bottom vertices (encircled in gray) represent the elements of $[w^d]$, which are fully connected as a clique, while each of the top vertices represents an $\atpc_{d,w}$ instance. Red edges correspond to the deletions while each of the blue edges is inserted but not deleted -- to avoid clutter, only the edges for $j=1$ are drawn. The green edge is $(1,n)$.}
	\label{fig:mst}
\end{figure}

At a high level, Alice and Bob jointly encode the target of each $\atpc_{d,w}$ instance as the minimum weight edge incident on a unique vertex.
To do this, the sender of a message (who does not own the current pointer) has no choice but to collect and merge its insertions/deletions from all subtrees (offset properly to make them disjoint).
On the other hand, the receiver owns the current pointer and thus has full knowledge of the subtrees to the left of the pointer, enabling perfect simulation of both parties in all these subtrees.
Suppose the receiver performs opposite operations on exactly the same subset of elements as the sender, meaning effectively no edge is inserted/deleted, in each of these subtrees.
As a result, the output of the larger instance always corresponds to the output of the smaller instance determined by the current pointer.
Besides, to ensure a proper inclusion, Alice as the receiver will insert everything to the right of the pointer while Bob as the receiver will delete nothing to the right of the pointer, as can be seen in the second case of~\Cref{alg:receiver}.

It can be verified that Alice is able to compute all insertions on her own and Bob is able to compute all deletions on his own.
So the $p$-pass dynamic streaming algorithm $\pi$ can be simulated by the protocol $\tau$, using $d$ rounds and $C$ communication, in the canonical way of exchanging memory states.
Now, it remains to formally show the correctness of the reduction.
This is done with the help of the following technical claim.

\begin{claim}\label{clm:sender-receiver}
	For any $\atpc_{d,w}$ instance $(A^{(d)},B^{(d)})$, it holds that
	\[
		\sender(d,w,A^{(d)},0) \supsetneq \receiver(d,w,B^{(d)},1),
	\]
	and
	\[
		\receiver(d,w,B^{(d)},0) \supsetneq \sender(d,w,A^{(d)},1).
	\]
	Furthermore, it also holds that
	\[
		\min(\sender(d,w,A^{(d)},0) \setminus \receiver(d,w,B^{(d)},1)) = 2t-z,
	\]
	and
	\[
		\min(\receiver(d,w,B^{(d)},0) \setminus \sender(d,w,A^{(d)},1)) = 2t-z,
	\]
	where $t$ is the target of the instance, and $z$ the output.
\end{claim}

Assume the above claim for now.
Applying it to $(a^{(d)}_j,b^{(d)}_j)$ for $j \in [k]$, we know that the constructed dynamic stream is well-defined and any minimum spanning tree of the constructed graph must consist of the following edges.
\begin{enumerate}
	\item The edge $(1,n)$ with weight $2kw^d-2T+1$.
	\item $w^d-1$ edges connecting $[k+1,k+w^d]$, each with weight $1$.
	\item The edge $(j,k+t_j)$ with weight $2t_j-z_j+1$, where $t_j$ is the target of $(a^{(d)}_j,b^{(d)}_j)$ and $z_j$ is the output of the same instance, for $j \in [k]$.
\end{enumerate}
Therefore, the weight of minimum spanning trees is
\[
	2kw^d-2T+1 + w^d-1 + \sum_{j \in [k]} (2t_j-z_j+1) = 2kw^d+w^d+k - \sum_{j \in [k]} z_j.
\]
In other words, $\tau$ will output $1$ (i.e., more than $\floor{k/2}$ out of the $k$ instances of $\atpc_{d,w}$ output $1$) if and only if $\pi$ outputs $0$ given threshold $2kw^d+w^d+k-\floor{k/2}$ (i.e., the weight of minimum spanning trees is less than the given threshold).
So the success probability remains the same.

We conclude this section with a proof of~\Cref{clm:sender-receiver}.

\begin{proof}[Proof of~\Cref{clm:sender-receiver}]
The proof is by induction on $d$.
The base case of $d=1$ is a direct consequence of the first cases of~\Cref{alg:sender} and~\Cref{alg:receiver}.
For $d>1$, we can get
\[
	\sender(d,w,A^{(d)},0) \supsetneq \receiver(d,w,B^{(d)},1),
\]
because by the inductive hypothesis, we have that
\[
	\receiver(d-1,w,b^{(d-1)}_{i^{(d)}},0) \supsetneq \sender(d-1,w,a^{(d-1)}_{i^{(d)}},1).
\]
Furthermore, we actually know
\begin{align*}
	& \min(\sender(d,w,A^{(d)},0) \setminus \receiver(d,w,B^{(d)},1))\\
	& \hspace{1cm} = 2(i^{(d)}-1) \cdot w^{d-1} + \min(\receiver(d-1,w,b^{(d-1)}_{i^{(d)}},0) \setminus \sender(d-1,w,a^{(d-1)}_{i^{(d)}},1))\\
	& \hspace{1cm} = 2(i^{(d)}-1) \cdot w^{d-1} + 2t'-z'\\
	& \hspace{1cm} = 2t-z,
\end{align*}
where $t'$ is the target of $(a^{(d-1)}_{i^{(d)}},b^{(d-1)}_{i^{(d)}})$ and $z'$ is the output of the same instance.

Similarly, we can also get
\[
	\receiver(d,w,B^{(d)},0) \supsetneq \sender(d,w,A^{(d)},1),
\]
because by the inductive hypothesis, we have that
\[
	\sender(d-1,w,a^{(d-1)}_{i^{(d)}},0) \supsetneq \receiver(d-1,w,b^{(d-1)}_{i^{(d)}},1).
\]
Furthermore, we actually know
\begin{align*}
	& \min(\receiver(d,w,B^{(d)},0) \setminus \sender(d,w,A^{(d)},1))\\
	& \hspace{1cm} = 2(i^{(d)}-1) \cdot w^{d-1} + \min(\sender(d-1,w,a^{(d-1)}_{i^{(d)}},0) \setminus \receiver(d-1,w,b^{(d-1)}_{i^{(d)}},1))\\
	& \hspace{1cm} = 2(i^{(d)}-1) \cdot w^{d-1} + 2t'-z'\\
	& \hspace{1cm} = 2t-z,
\end{align*}
where $t'$ is the target of $(a^{(d-1)}_{i^{(d)}},b^{(d-1)}_{i^{(d)}})$ and $z'$ is the output of the same instance.
This concludes the proof.
\end{proof}

We remark that~\Cref{clm:sender-receiver} can also be viewed as a reduction from $\atpc_{d,w}$ to $\dmur$ over a universe of size $m=2w^d$, by Alice computing the set $\sender(d,w,A^{(d)},0)$ and Bob computing the set $\receiver(d,w,B^{(d)},1)$.
So the following corollary follows immediately from the lower bound for $\atpc_{d,w}$ (\Cref{lem:lb-atpc}).

\begin{corollary}\label{cor:dmur}
	For $m,r \ge 1$, and $\epsilon \in [0,1/2]$, any $r$-round (randomized) protocol that solves $\dmur$ with probability $1/2+\epsilon$ over a universe of size $m$, requires $\Omega(\epsilon^2 m^{1/r}/r)$ communication.
\end{corollary}

\subsection{Lower Bound for $\atpc_{d,w}$}\label{sec:proof-lb-atpc}

We derive~\Cref{lem:lb-atpc} by round elimination in this section.
\Cref{clm:base-case,clm:round-elim} take care of the base case and each round elimination step, respectively.

\begin{claim}[Base case]\label{clm:base-case}
	For $w \ge 1$, any one-way protocol $\pi$ for solving $\atpc_{1,w}$ succeeds with probability at most $1/2+\sqrt{\cc(\pi)/w}$ over $\cD_{1,w}$.
\end{claim}

\begin{proof}
	Throughout the proof, all superscripts on random variables will be temporarily omitted for conciseness.
	Let $\rM$ be the message sent from Alice to Bob.
	We can get
	\begin{align*}
		\mi{\rM,\rA_{< \rI},\rI}{\rA_\rI}
		& = \mi{\rA_{< \rI},\rI}{\rA_\rI} + \mi{\rM}{\rA_\rI \mid \rA_{< \rI},\rI} \tag{by chain rule of mutual information (\itfacts{chain-rule})}\\
		& = \mi{\rM}{\rA_\rI \mid \rA_{< \rI},\rI} \tag{as $\rA_{< \rI},\rI \perp \rA_\rI$}\\
		& = \frac{1}{w} \cdot \sum_{i \in [w]} \mi{\rM}{\rA_i \mid \rA_{< i},\rI=i} \tag{as $\rI$ is uniform}\\
		& = \frac{1}{w} \cdot \sum_{i \in [w]} \mi{\rM}{\rA_i \mid \rA_{< i}} \tag{as $\rM,\rA_{\le i} \perp \rI=i$}\\
		& = \frac{\mi{\rM}{\rA}}{w} \tag{by chain rule of mutual information (\itfacts{chain-rule})}\\
		& \le \frac{\cc(\pi)}{w}.
	\end{align*}
	Furthermore, we have
	\begin{align*}
		& \Exp_{\rM,\rA_{< \rI},\rI} \tvd{\distribution{\rA_\rI \mid \rM,\rA_{< \rI},\rI}}{\distribution{\rA_\rI}}\\
		& \hspace{1cm} \le \Exp_{\rM,\rA_{< \rI},\rI} \sqrt{\kl{\distribution{\rA_\rI \mid \rM,\rA_{< \rI},\rI}}{\distribution{\rA_\rI}}} \tag{by Pinsker's inequality (\Cref{fact:pinskers})}\\
		& \hspace{1cm} \le \sqrt{\Exp_{\rM,\rA_{< \rI},\rI} \kl{\distribution{\rA_\rI \mid \rM,\rA_{< \rI},\rI}}{\distribution{\rA_\rI}}} \tag{by concavity of $\sqrt{\cdot}$}\\
		& \hspace{1cm} = \sqrt{\mi{\rA_{\rI}}{\rM,\rA_{< \rI},\rI}} \tag{by~\Cref{fact:kl-info}}\\
		& \hspace{1cm} \le \sqrt{\frac{\cc(\pi)}{w}}.
	\end{align*}
	Since $\rA_\rI$ is initially uniformly random, the probability that Bob can correctly guess $\rA_\rI$ from his perspective (i.e., given $\rM,\rA_{< \rI},\rI$) is then at most $1/2+\sqrt{\cc(\pi)/w}$, as claimed.
\end{proof}

\begin{claim}[Round elimination]\label{clm:round-elim}
	For $d,w \ge 1$ and $\epsilon \in [0,1]$, if there exists a $(d+1)$-round protocol $\pi$ for solving $\atpc_{d+1,w}$ with probability $\epsilon$ over $\cD_{d+1,w}$, then there also exists a $d$-round protocol $\tau$ for solving $\atpc_{d,w}$ with probability $\epsilon-\sqrt{\cc^{(1)}(\pi)/w}$ over $\cD_{d,w}$ and communication $\cc(\tau) = \cc^{(> 1)}(\pi)$.
\end{claim}

\begin{proof}
	Throughout the proof, all superscripts on random variables will be temporarily omitted for conciseness.
	When there is ambiguity, unprimed quantities represent the $d$-round versions while primed ones are reserved for $d+1$ rounds.
	Let $\rM$ be the first-round message of $\pi$.
	On input $(A,B)$, $\tau$ will simulate $\pi$ as shown in~\Cref{alg:round-elim}, where all random variables are with respect to $\pi$.

	\begin{Algorithm}\label{alg:round-elim}
		The $d$-round protocol $\tau$ for solving $\atpc_{d,w}$ on input $(A,B)$.
		\begin{enumerate}
			\item Alice and Bob publicly sample $\rM,\rB_{< \rI},\rI$.
			\item Alice sets $\rA_\rI = A$ and Bob sets $\rB_\rI = B$.
			\item Alice privately samples $\rA_{< \rI}$ conditioned on $\rM,\rA_\rI,\rB_{< \rI},\rI$, and sets $\rB' = (\rI,\rA_{\le \rI},\rB_{< \rI})$.
			\item Bob privately samples $\rB_{> \rI}$ conditioned on $\rM,\rB_{\le \rI},\rI$, and sets $\rA' = \rB$.
			\item Alice and Bob simulate $\pi$ on input $(\rB',\rA')$ from the second round, assuming the first-round message is $\rM$, with Alice playing the role of Bob and Bob playing the role of Alice.
			\item Return the output of $\pi$.
		\end{enumerate}
	\end{Algorithm}

	It can be verified that $(\rA',\rB')$ is a valid instance of $\atpc_{d+1,w}$ and thus $\tau$ is indeed a $d$-round protocol for solving $\atpc_{d,w}$, with the claimed communication complexity.
	We also want to emphasize that with Alice and Bob playing the roles of each other, $\rM$ is the first-round message of $\pi$ from Bob to Alice, so it is a function of $\rA'=\rB$.
	Now, it remains to calculate the success probability of $\tau$.
	To this end, the following two technical claims show that all the random variables as sampled in $\tau$ almost perfectly follow their distribution in $\pi$.

	\begin{claim}\label{clm:mi-public}
		It holds that
		\[
			\mi{\rA_\rI,\rB_\rI}{\rM,\rB_{< \rI},\rI} \le \frac{\cc^{(1)}(\pi)}{w}.
		\]
	\end{claim}

	\begin{proof}
		Observe that
		\begin{align*}
			\mi{\rA_\rI,\rB_\rI}{\rM,\rB_{< \rI},\rI}
			& = \mi{\rA_\rI,\rB_\rI}{\rB_{< \rI},\rI} + \mi{\rA_\rI,\rB_\rI}{\rM \mid \rB_{< \rI},\rI} \tag{by chain rule of mutual information (\itfacts{chain-rule})}\\
			& = \mi{\rA_\rI,\rB_\rI}{\rM \mid \rB_{< \rI},\rI} \tag{as $\rA_\rI,\rB_\rI \perp \rB_{< \rI},\rI$}\\
			& = \mi{\rB_\rI}{\rM \mid \rB_{< \rI},\rI} + \mi{\rA_\rI}{\rM \mid \rB_{\le \rI},\rI} \tag{by chain rule of mutual information (\itfacts{chain-rule})}\\
			& \le \mi{\rB_\rI}{\rM \mid \rB_{< \rI},\rI} + \mi{\rA_\rI}{\rB_{> \rI} \mid \rB_{\le \rI},\rI} \tag{by data processing inequality (\itfacts{data-processing}) as $\rM$ is a function of $\rB_{> \rI},\rB_{\le \rI}$}\\
			& = \mi{\rB_\rI}{\rM \mid \rB_{< \rI},\rI} \tag{as $\rA_\rI \perp \rB_{> \rI} \mid \rB_{\le \rI},\rI$}\\
			& = \frac{1}{w} \cdot \sum_{i \in [w]} \mi{\rB_i}{\rM \mid \rB_{< i},\rI=i} \tag{as $\rI$ is uniform}\\
			& = \frac{1}{w} \cdot \sum_{i \in [w]} \mi{\rB_i}{\rM \mid \rB_{< i}} \tag{as $\rM,\rB_{\le i} \perp \rI=i$}\\
			& = \frac{\mi{\rB}{\rM}}{w} \tag{by chain rule of mutual information (\itfacts{chain-rule})}\\
			& \le \frac{\cc^{(1)}(\pi)}{w},
		\end{align*}
		as claimed.
	\end{proof}

	\begin{claim}\label{clm:mi-private}
		It holds that
		\[
			\rA_{< \rI} \perp \rB_\rI \mid \rM,\rA_\rI,\rB_{< \rI},\rI,
		\]
		and
		\[
			\rB_{> \rI} \perp \rA_{\le \rI} \mid \rM,\rB_{\le \rI},\rI.
		\]
	\end{claim}

	\begin{proof}
		Observe that
		\begin{align*}
			\mi{\rA_{< \rI}}{\rB_\rI \mid \rM,\rA_\rI,\rB_{< \rI},\rI}
			& \le \mi{\rA_{< \rI}}{\rM,\rB_\rI \mid \rA_\rI,\rB_{< \rI},\rI}\\
			& \le \mi{\rA_{< \rI}}{\rB_{\ge \rI} \mid \rA_\rI,\rB_{< \rI},\rI} \tag{by data processing inequality (\itfacts{data-processing}) as $\rM$ is a function of $\rB_{\ge \rI},\rB_{< \rI}$}\\
			& = 0, \tag{as $\rA_{< \rI} \perp \rB_{\ge \rI} \mid \rA_\rI,\rB_{< \rI},\rI$}
		\end{align*}
		and that
		\begin{align*}
			\mi{\rB_{> \rI}}{\rA_{\le \rI} \mid \rM,\rB_{\le \rI},\rI}
			& \le \mi{\rM,\rB_{> \rI}}{\rA_{\le \rI} \mid \rB_{\le \rI},\rI}\\
			& \le \mi{\rB_{> \rI}}{\rA_{\le \rI} \mid \rB_{\le \rI},\rI} \tag{by data processing inequality (\itfacts{data-processing}) as $\rM$ is a function of $\rB_{> \rI},\rB_{\le \rI}$}\\
			& = 0. \tag{as $\rB_{> \rI} \perp \rA_{\le \rI} \mid \rB_{\le \rI},\rI$}
		\end{align*}
		This concludes the proof.
	\end{proof}

	Note that $\tau$ succeeds on $(\rA_\rI,\rB_\rI)$ so long as $\pi$ does on $(\rA',\rB')$.
	Intuitively, $\tau$ samples $(\rA',\rB')$ closely following $\cD_{d+1,w}$ so its success probability should also be close to the success probability of $\pi$.
	This can be formally argued as follows.
	By chain rule of total variation distance (\Cref{fact:tvd-chain-rule}), together with~\Cref{clm:mi-public,clm:mi-private}, we can get that the distribution sampled by $\tau$ of all random variables ($\rM,\rA_{\le \rI},\rB,\rI$) and their distribution in $\pi$ have a total variation distribution upper bounded by
	\begin{align*}
		& \Exp_{\rM,\rB_{< \rI},\rI} \tvd{\distribution{\rA_\rI,\rB_\rI \mid \rM,\rB_{< \rI,\rI}}}{\distribution{\rA_\rI,\rB_\rI}}\\
		& \hspace{1cm} \le \Exp_{\rM,\rB_{< \rI},\rI} \sqrt{\kl{\distribution{\rA_\rI,\rB_\rI \mid \rM,\rB_{< \rI,\rI}}}{\distribution{\rA_\rI,\rB_\rI}}} \tag{by Pinsker's inequality (\Cref{fact:pinskers})}\\
		& \hspace{1cm} \le \sqrt{\Exp_{\rM,\rB_{< \rI},\rI} \kl{\distribution{\rA_\rI,\rB_\rI \mid \rM,\rB_{< \rI,\rI}}}{\distribution{\rA_\rI,\rB_\rI}}} \tag{by concavity of $\sqrt{\cdot}$}\\
		& \hspace{1cm} = \sqrt{\mi{\rA_\rI,\rB_\rI}{\rM,\rB_{< \rI,\rI}}} \tag{by~\Cref{fact:kl-info}}\\
		& \hspace{1cm} \le \sqrt{\frac{\cc^{(1)}(\pi)}{w}}.
	\end{align*}
	As a result, the overall success probability of $\tau$ is less than that of $\pi$ by at most $\sqrt{\cc^{(1)}(\pi)/w}$ due to~\Cref{fact:tvd-small}, concluding the proof.
\end{proof}

Now, we are ready to prove~\Cref{lem:lb-atpc}.

\begin{proof}[Proof of~\Cref{lem:lb-atpc}]
	Fix a $d$-round protocol $\pi$ for solving $\atpc_{d,w}$ with probability $1/2+\epsilon$ over $\cD_{d,w}$ that has optimal communication.
	By applying the round elimination step (\Cref{clm:round-elim}) repeatedly for $d-1$ times, we get a one-way protocol for solving $\atpc_{1,w}$ with probability $1/2+\epsilon - \sum_{r \in [d-1]} \sqrt{\cc^{(r)}(\pi)/w}$ over $\cD_{1,w}$ and communication $\cc^{(d)}(\pi)$.
	On the other hand, \Cref{clm:base-case} implies that
	\begin{align*}
		\epsilon
		& \le \sum_{r \in [d]} \sqrt{\frac{\cc^{(r)}(\pi)}{w}}\\
		& = d \cdot \sum_{r \in [d]} \frac{1}{d} \cdot \sqrt{\frac{\cc^{(r)}(\pi)}{w}}\\
		& \le d \cdot \sqrt{\sum_{r \in [d]} \frac{1}{d} \cdot \frac{\cc^{(r)}(\pi)}{w}} \tag{by concavity of $\sqrt{\cdot}$}\\
		& = \sqrt{\frac{d}{w} \cdot \cc(\pi)}.
	\end{align*}
	The lower bound is derived by rearranging the terms.
\end{proof}

\subsection{Lower Bound for $\hintatpc^{\# k}_{d,w}$}\label{sec:proof-lb-hint-atpc-maj}

We generalize the lower bound for $\atpc_{d,w}$ to its $k$-fold verisons in this section.
The hardness of $\hintatpc^{\# k}_{d,w}$ will be proved by a series of reductions from $\atpc^{\oplus k}_{d,w}$.
Indeed, \Cref{clm:atpc-xor-hint-atpc-xor} deals with the hint, namely the sum of targets, while \Cref{clm:hint-atpc-xor-hint-atpc-maj} relates the XOR and the majority versions of the problem.
The hardness of $\atpc^{\oplus k}_{d,w}$ will in turn be derived from the hardness of $\atpc_{d,w}$ using the XOR lemma of~\cite{Yu22} (\Cref{lem:xor-yu}).

\begin{claim}\label{clm:atpc-xor-hint-atpc-xor}
	For $k,d,w \ge 1$ and $\epsilon \in [0,1/2]$, it holds that
	\[
		\bD^{(d)}_{\cD^k_{d,w},\frac{1}{2}+\frac{\epsilon}{kw^d}}(\atpc^{\oplus k}_{d,w}) \le \bD^{(d)}_{\cD^k_{d,w},\frac{1}{2}+\epsilon}(\hintatpc^{\oplus k}_{d,w}) + kd\log w.
	\]
\end{claim}

\begin{proof}
	Fix a $d$-round protocol $\pi$ for solving $\hintatpc^{\oplus k}_{d,w}$ with probability $1/2+\epsilon$ over $\cD^k_{d,w}$ that has optimal communication.
	In the following, we construct a protocol $\tau$ for solving $\atpc^{\oplus k}_{d,w}$ with the claimed properties.
	On input $(A,B)$, $\tau$ starts by publicly guessing $T'$ from $[kw^d]$ uniformly at random and simulating $\pi$ on input $((A,T'),(B,T'))$.
	Then, Alice sends the first message of $\pi$ to Bob.
	For $r \in [2,d]$, the sender of the $r$-th round will transmit the $r$-th message of $\pi$, which can be simulated using the current transcript, as well as the correct depth-$(r-1)$ pointers for each of the $k$ instances of $\atpc_{d,w}$.
	So the receiver of the $r$-th round will know what the correct depth-$r$ pointers are.
	At the end of the protocol, the receiver of the last message has full knowledge required to compute the sum $T$ of all $k$ targets and can verify whether $T'=T$.
	If so, return the output of $\pi$, and otherwise a uniformly random bit.

	Note that apart from messages of $\pi$, $\tau$ transmits $k$ pointers, each of which costs $\log w$ bits, in each of the $d$ rounds, so the additional communication is $kd\log w$ in total.
	To see its success probability, observe that $T'$ is a uniformly random guess independent of $(A,B)$.
	Thus, it correctly hits $T$ with probability $1/(kw^d)$ and $\tau$ succeeds with probability $1/2+\epsilon$, the same as $\pi$, conditioned on this event.
	Otherwise, $\tau$ succeeds with probability $1/2$ as a uniformly random bit is output.
	The overall success probability then follows as claimed.
\end{proof}

\begin{claim}\label{clm:hint-atpc-xor-hint-atpc-maj}
	For $k,d,w \ge 1$ and $\epsilon \in [0,1/2]$, it holds that
	\[
		\bD^{(d)}_{\cD^k_{d,w},\frac{1}{2}-\epsilon+\Theta(\frac{1}{\sqrt{k}})}(\hintatpc^{\oplus k}_{d,w}) \le \bD^{(d)}_{\cD^k_{d,w},1-\epsilon}(\hintatpc^{\# k}_{d,w}).
	\]
\end{claim}

\begin{proof}
	Fix a $d$-round protocol $\pi$ for solving $\hintatpc^{\# k}_{d,w}$ with probability $1-\epsilon$ over $\cD^k_{d,w}$ that has optimal communication.
	In the following, we construct a protocol $\tau$ for solving $\hintatpc^{\oplus k}_{d,w}$ with the claimed properties.
	On input $((A,T),(B,T))$, $\tau$ simulates $\pi$ on input $((A,T),(B,T))$, and outputs the parity of $\floor{k/2}+1$ if $\pi$ outputs $1$ and outputs the parity of $\floor{k/2}$ otherwise.
	So it remains to calculate the success probability of $\tau$.

	Assume for now that $\pi$ were perfectly correct (i.e., $\epsilon=0$).
	Note that the answer to an instance drawn from $\cD_{d,w}$ is simply a uniformly random bit.
	Therefore, $\tau$ succeeds with probability
	\begin{align*}
		& \sum_{j=0}^{\floor{k/2}} \mathbf{1}[\text{$\floor{k/2}-j$ is even}] \cdot \frac{\binom{k}{j}}{2^k} + \sum_{j=\floor{k/2}+1}^k \mathbf{1}[\text{$j-\floor{k/2}-1$ is even}] \cdot \frac{\binom{k}{j}}{2^k}\\
		& \hspace{1cm} = \frac{1}{2} \cdot \bracket{\sum_{j=0}^{\floor{k/2}} \frac{\binom{k}{j}}{2^k} + \sum_{j=0}^{\floor{k/2}} (-1)^{\floor{k/2}-j} \cdot \frac{\binom{k}{j}}{2^k}}\\
		& \hspace{2cm} + \frac{1}{2} \cdot \bracket{\sum_{j=\floor{k/2}+1}^k \frac{\binom{k}{j}}{2^k} + \sum_{j=\floor{k/2}+1}^k (-1)^{j-\floor{k/2}-1} \cdot \frac{\binom{k}{j}}{2^k}}\\
		& \hspace{1cm} = \frac{1}{2} \cdot \bracket{1 + \sum_{j=0}^{\floor{k/2}} (-1)^{\floor{k/2}-j} \cdot \frac{\binom{k}{j}}{2^k} + \sum_{j=\floor{k/2}+1}^k (-1)^{j-\floor{k/2}-1} \cdot \frac{\binom{k}{j}}{2^k}}\\
		& \hspace{1cm} = \frac{1}{2} \cdot \bracket{1 + \frac{\binom{k-1}{\floor{k/2}}}{2^k} + \frac{\binom{k-1}{\floor{k/2}}}{2^k}}\\
		& \hspace{1cm} = \frac{1}{2} + \frac{\binom{k-1}{\floor{k/2}}}{2^k}\\
		& \hspace{1cm} = \frac{1}{2} + \Theta\paren{\frac{1}{\sqrt{k}}},
	\end{align*}
	by Stirling's approximation.
	For $\pi$ with error probability $\epsilon$, by a union bound, the success probability of $\tau$ is reduced by at most $\epsilon$.
	The claim hence follows.
\end{proof}

We remark that the reductions in proving~\Cref{clm:atpc-xor-hint-atpc-xor,clm:hint-atpc-xor-hint-atpc-maj} actually have little to do with the base problem $\atpc_{d,w}$ itself.
In fact, almost identical reductions will also be used in~\Cref{sec:optimal} for proving the full version of our main result.

Now, we are ready to prove~\Cref{lem:lb-hint-atpc-maj} using the following XOR lemma of~\cite{Yu22}\footnote{Technically, the main result (Theorem 1) of~\cite{Yu22} is an XOR lemma for randomized communication complexity, while a distributional version is required in our proof. Nevertheless, \cite{Yu22} proves the main result via another one (Theorem 2) with asymmetirc communication, which directly works in the distributional model. The distributional version we need is a natural byproduct of the simple argument from Theorem 2 to Theorem 1; see Section 4 in the full version of~\cite{Yu22} for more details.}.

\begin{lemma}[\cite{Yu22}]\label{lem:xor-yu}
	For $k,r \ge 1$, Boolean function $f$, and input distance $\mu$, it holds that
	\[
		\bD^{(r)}_{\mu^k,\frac{1}{2}+\frac{1}{2^k}}(f^{\oplus k}) \ge k \cdot \paren{\frac{1}{r^{O(r)}} \cdot \bD^{(r)}_{\mu,\frac{2}{3}}(f) - 1}.
	\]
\end{lemma}

\begin{proof}[Proof of~\Cref{lem:lb-hint-atpc-maj}]
	Fix a $d$-round protocol for solving $\hintatpc^{\# k}_{d,w}$ with probability $1-1/\poly(k)$ over $\cD^k_{d,w}$ that has optimal communication $C$.
	Applying~\Cref{clm:hint-atpc-xor-hint-atpc-maj} and~\Cref{clm:atpc-xor-hint-atpc-xor} in sequence, we get a $d$-round protocol for solving $\atpc^{\oplus k}_{d,w}$ with probability $1/2+\Theta(1/(k^{3/2}w^d))$ over $\cD^k_{d,w}$ that has communication $C+kd\log w$.

	On the other hand, \Cref{lem:lb-atpc}, together with \Cref{lem:xor-yu}, implies that
	\[
		\bD^{(d)}_{\cD^k_{d,w},\frac{1}{2}+\frac{1}{2^k}}(\atpc^{\oplus k}_{d,w}) = \Omega\paren{\frac{kw}{d^{O(d)}}},
	\]
	under the assumption $d = o(\log w / \log\log w)$.
	Since $1/2^k \ll 1/(k^{3/2}w^d)$ under the assumption $k = \omega(d\log w)$, we have that $C = \Omega(kw/d^{O(d)})$.
	The lemma follows by observing that error reduction from $1/3$ down to $1/\poly(k)$ requires an $O(\log k)$-fold increase in communication.
\end{proof}


\section{A Lower Bound in Optimal Number of Passes}\label{sec:optimal}

In this section, we extend~\Cref{thm:lb-few} to more passes as shown in~\Cref{thm:lb-optimal}, fully proving~\Cref{res:main}.
Also, \Cref{res:xor}, formalized in~\Cref{thm:xor}, will be a direct consequence of~\Cref{lem:xor-direct-sum,lem:multi-xor}.

\begin{theorem}[Formal version of~\Cref{res:main}]\label{thm:lb-optimal}
	For $p = o(\frac{\log n}{\log\log n})$, any $p$-pass dynamic streaming algorithm for solving $\mst_n$ with probability $2/3$ requires $\Omega(\frac{n^{1+\frac{1}{2p-1}}}{p^5 \log^3 n})$ space.
\end{theorem}

\begin{theorem}[Formal version of~\Cref{res:xor}]\label{thm:xor}
	There exists $\epsilon_0 > 0$ such that for $n,r \ge 1$, $k \in [1,n]$, $\epsilon \in (0,\epsilon_0)$, Boolean function $f$, and input distribution $\mu$, it holds that
	\[
		\bD^{(r),k}_{\mu^n,\frac{1}{2}+\min(\epsilon_1,\epsilon_2)}(f^{\oplus n}) = \Omega\paren{\frac{n}{k} \cdot \paren{\frac{\epsilon}{r} \cdot \bD^{(r)}_{\mu,\frac{1}{2}+\epsilon}(f) - O(r)}},
	\]
	where $\epsilon_1 = (r\epsilon)^{\Omega(k/r)}$ and $\epsilon_2 = \epsilon^{\Omega(\epsilon k/r)}$.
\end{theorem}

As mentioned in~\Cref{sec:few}, we take a two-step approach by amplifying first communication and then error probability.
The first step uses the following XOR-direct-sum result\footnote{A similar direct-sum result also holds for $f^k$. It is however subsumed by the direct-product result of~\cite{JainPY12}.}, tight up to a factor of $r$, for bounded-round communication complexity, which is implicitly implied by~\cite{JainPY12}.

\begin{lemma}[Bounded-round XOR direct sum]\label{lem:xor-direct-sum}
	For $k,r \ge 1$, $\epsilon \in [0,1]$, $\delta \in (0,\epsilon)$, Boolean function $f$, and input distribution $\mu$, it holds that
	\[
		\bD^{(r)}_{\mu^k,\epsilon}(f^{\oplus k}) = \Omega\paren{k \cdot \paren{\frac{\delta}{r} \cdot \bD^{(r)}_{\mu,\epsilon-\delta}(f) - O(r)}}.
	\]
\end{lemma}

For the second step, we prove an XOR-lemma-type result for the multi-party model.

\begin{lemma}[Multi-party XOR lemma]\label{lem:multi-xor}
	There exists $\epsilon_0 > 0$ such that for $k,r \ge 1$, $\epsilon \in (0,\epsilon_0)$, Boolean function $f$, and input distribution $\mu$, it holds that
	\[
		\bD^{(r),k}_{\mu^k,\frac{1}{2}+\min(\epsilon_1,\epsilon_2)}(f^{\oplus k}) \ge \bD^{(r)}_{\mu,\frac{1}{2}+\epsilon}(f),
	\]
	where $\epsilon_1 = (r\epsilon)^{\Omega(k/r)}$ and $\epsilon_2 = \epsilon^{\Omega(\epsilon k/r)}$.
\end{lemma}

We remark that~\Cref{lem:multi-xor} is a weaker XOR lemma than the one of~\cite{Yu22} in the sense that the decrease in advantage is worse and it only applies to the multi-party model.
It nevertheless meets our needs as we will eventually work in the streaming model.
On the positive side, \Cref{lem:multi-xor} no longer suffers a factor of $r^{O(r)}$ in communication, which is exactly the only barrier towards $o(\log n / \log\log n)$ passes as identified in~\Cref{sec:few}.

Formal proofs of~\Cref{lem:xor-direct-sum,lem:multi-xor} are deferred to~\Cref{sec:proof}.
We now show that they indeed imply~\Cref{thm:lb-optimal}.
This is done via multi-party variants of the problems in~\Cref{sec:few}.

\begin{definition}
	For $k_1,k_2,d,w \ge 1$, $\atpc^{\oplus (k_1,k_2)}_{d,w}$ denotes the $k_1 k_2$-fold XOR version of $\atpc_{d,w}$ in the $2k_2$-party model, where each pair of Alice $i$ and Bob $i$ is given as input $k_1$ instances of $\atpc_{d,w}$, where $i \in [k_2]$.
	Similarly, $\atpc^{\# (k_1,k_2)}_{d,w}$ denotes the $k_1 k_2$-fold majority version of $\atpc_{d,w}$ in the $2k_2$-party model.

	When the hint, i.e., the total sum of targets of all $k_1 k_2$ instances, is given to each of the $2k_2$ parties as part of its input, the resulting problems are denoted by $\hintatpc^{\oplus (k_1,k_2)}_{d,w}$ and $\hintatpc^{\# (k_1,k_2)}_{d,w}$.
\end{definition}

Multi-party analogues of~\Cref{clm:atpc-xor-hint-atpc-xor,clm:hint-atpc-xor-hint-atpc-maj,clm:hint-atpc-maj-mst} are given below.
The proofs are almost identical and omitted here, as the same reductions still apply.

\begin{claim}\label{clm:multi-hint-atpc-maj-mst}
	For $k_1,k_2,p,w,S \ge 1$ and $\epsilon \in [0,1]$, if there exists a $p$-pass, $S$-space dynamic streaming algorithm for solving $\mst_{k_1 k_2 + w^{2p-1} + 1}$ with probability $\epsilon$, then there also exists a $(2p-1)$-round, $(2p-1)k_2 S$-communication protocol for solving $\hintatpc^{\# (k_1,k_2)}_{2p-1,w}$ with probability $\epsilon$ over $\cD^{k_1 k_2}_{2p-1,w}$.
\end{claim}

\begin{claim}\label{clm:multi-atpc-xor-hint-atpc-xor}
	For $k_1,k_2,d,w \ge 1$, and $\epsilon \in [0,1/2]$, it holds that
	\[
		\bD^{(d),k_2}_{\cD^{k_1 k_2}_{d,w},\frac{1}{2}+\frac{\epsilon}{k_1 k_2 w^{d}}}(\atpc^{\oplus (k_1,k_2)}_{d,w}) \le \bD^{(d),k_2}_{\cD^{k_1 k_2}_{d,w},\frac{1}{2}+\epsilon}(\hintatpc^{\oplus (k_1,k_2)}_{d,w}) + k_1 k_2 d\log w.
	\]
\end{claim}

\begin{claim}\label{clm:multi-hint-atpc-xor-hint-atpc-maj}
	For $k_1,k_2,d,w \ge 1$ and $\epsilon \in [0,1/2]$, it holds that
	\[
		\bD^{(d),k_2}_{\cD^{k_1 k_2}_{d,w},\frac{1}{2}-\epsilon+\Theta(\frac{1}{\sqrt{k_1 k_2}})}(\hintatpc^{\oplus (k_1,k_2)}_{d,w}) \le \bD^{(d),k_2}_{\cD^{k_1 k_2}_{d,w},1-\epsilon}(\hintatpc^{\# (k_1,k_2)}_{d,w}).
	\]
\end{claim}

We are now ready to prove~\Cref{thm:lb-optimal}.

\begin{proof}[Proof of~\Cref{thm:lb-optimal}]
	Fix a $p$-pass dynamic streaming algorithm for solving $\mst_n$ with probability $2/3$ that has space $S$.
	Let $k=(n-1)/2$, $k_2 = c d\log n$ and $k_1=k/k_2$ for some sufficiently large constant $c > 0$.
	Also let $d=2p-1$, $w=(n-k-1)^{1/d}$, and $C=k_2 dS$.
	Applying the reduction of~\Cref{clm:multi-hint-atpc-maj-mst}, we get a $d$-round protocol for solving $\hintatpc^{\# (k_1,k_2)}_{d,w}$ with probability $2/3$ over $\cD^k_{d,w}$ that has communication $C$.
	The success probability can be boosted to $1-1/\poly(n)$ by $O(\log n)$ parallel repetitions.

	On the other hand, \Cref{lem:lb-atpc}, together with~\Cref{thm:xor} for some sufficiently small constant $\epsilon > 0$, implies that
	\[
		\bD^{(d),k_2}_{\cD^k_{d,w},\frac{1}{2}+\frac{1}{\poly(n)}}(\atpc^{\oplus (k_1,k_2)}_{d,w}) = \Omega\paren{\frac{k_1 w}{d^2}}.
	\]
	Applying~\Cref{clm:multi-atpc-xor-hint-atpc-xor,clm:multi-hint-atpc-xor-hint-atpc-maj} in sequence,
	we further get
	\[
		\bD^{(d),k_2}_{\cD^k_{d,w},1-\frac{1}{\poly(n)}}(\hintatpc^{\# (k_1,k_2)}_{d,w}) = \Omega\paren{\frac{k_1 w}{d^2}}.
	\]
	Combining the above arguments, we finally have
	\[
		C\log n = \Omega\paren{\frac{k_1 w}{d^2}}.
	\]
	The theorem follows by rearranging the terms.
\end{proof}

We remark that the above proof uses the $\epsilon_2$ case in~\Cref{lem:multi-xor} with $\epsilon=\Theta(1)$.
It is also possible to prove~\Cref{thm:lb-optimal} using the $\epsilon_1$ case with $\epsilon=\Theta(1/r)$.
However, this results in slightly worse dependence on $p$ for the derived space lower bound on streaming algorithms.
Both cases of~\Cref{lem:multi-xor} are provided just in case the result may be of independent interest to some readers.

\subsection*{Acknowledgement}

Sepehr Assadi is supported in part by an Alfred P. Sloan Fellowship, a University of Waterloo startup grant, an NSF CAREER grant CCF-2047061, and a gift from Google Research.
Gillat Kol is supported by a National Science Foundation CAREER award CCF-1750443 and by a BSF grant No. 2018325.

\clearpage

\bibliographystyle{alpha}
\bibliography{general}
\clearpage

\appendix

\part*{Appendix}


\section{Basic Tools From Information Theory}\label{sec:info}

We introduce some definitions from information theory that are needed in this paper.
For a random variable $\rA$, we use $\supp{\rA}$ to denote the support of $\rA$ and $\distribution{\rA}$ to denote its distribution.
When it is clear from context, we may abuse the notation and use $\rA$ directly instead of $\distribution{\rA}$, e.g., write
$A \sim \rA$ to mean $A \sim \distribution{\rA}$, i.e., $A$ is sampled from the distribution of the random variable $\rA$.

We denote the \emph{Shannon entropy} of a random variable $\rA$ by
$\en{\rA}$, which is defined as:
\begin{align*}
\en{\rA} = \sum_{A \in \supp{\rA}} \Pr\paren{\rA = A} \cdot \log\frac{1}{\Pr\paren{\rA = A}}.
\end{align*}
The \emph{conditional entropy} of $\rA$ conditioned on $\rB$ is denoted by $\en{\rA \mid \rB}$ and defined as:
\begin{align*}
\en{\rA \mid \rB} = \Ex_{B \sim \rB} \bracket{\en{\rA \mid \rB = B}},
\end{align*}
where $\en{\rA \mid \rB = B}$ is defined in a standard way by using the distribution of $\rA$ conditioned on the event $\rB = B$ in the previous equation.
We denote the \emph{mutual information} between two random variables $\rA$ and $\rB$ is by
$\mi{\rA}{\rB}$, which is defined as:
\begin{align*}
\mi{\rA}{\rB} = \en{\rA} - \en{\rA \mid  \rB} = \en{\rB} - \en{\rB \mid  \rA}.
\end{align*}
The \emph{conditional mutual information} $\mi{\rA}{\rB \mid \rC}$ is defined to be $\en{\rA \mid \rC} - \en{\rA \mid \rB,\rC}$ and hence by linearity of expectation:
\begin{align*}
\mi{\rA}{\rB \mid \rC} = \Ex_{C \sim \rC} \bracket{\mi{\rA}{\rB \mid \rC = C}}.
\end{align*}

We refer the interested readers to the excellent textbook by Cover and Thomas~\cite{CoverT06} for an introduction to the field of information theory.

\subsection{Useful Properties of Entropy and Mutual Information}\label{sec:prop-en-mi}

We use the following basic properties of entropy and mutual information throughout.

\begin{fact}[cf.~\cite{CoverT06}]\label{fact:it-facts}
  Let $\rA$, $\rB$, $\rC$, and $\rD$ be four (possibly correlated) random variables.
   \begin{enumerate}
  \item \label{part:uniform} $0 \leq \en{\rA} \leq \log{\card{\supp{\rA}}}$. The right equality holds
    iff $\distribution{\rA}$ is uniform.
  \item \label{part:info-zero} $\mi{\rA}{\rB}[\rC] \geq 0$. The equality holds iff $\rA$ and
    $\rB$ are \emph{independent} conditioned on $\rC$.
  \item \label{part:cond-reduce} \emph{Conditioning on a random variable reduces entropy}:
    $\en{\rA \mid \rB,\rC} \leq \en{\rA \mid  \rB}$.  The equality holds iff $\rA \perp \rC \mid \rB$.
    \item \label{part:sub-additivity} \emph{Subadditivity of entropy}: $\en{\rA,\rB \mid \rC}
    \leq \en{\rA \mid C} + \en{\rB \mid  \rC}$.
   \item \label{part:ent-chain-rule} \emph{Chain rule for entropy}: $\en{\rA,\rB \mid \rC} = \en{\rA \mid \rC} + \en{\rB \mid \rC,\rA}$.
  \item \label{part:chain-rule} \emph{Chain rule for mutual information}: $\mi{\rA,\rB}{\rC \mid \rD} = \mi{\rA}{\rC \mid \rD} + \mi{\rB}{\rC \mid  \rA,\rD}$.
  \item \label{part:data-processing} \emph{Data processing inequality}: for a function $f(\rA,\rC)$, $\mi{f(\rA,\rC)}{\rB \mid \rC} \leq \mi{\rA}{\rB \mid \rC}$.
   \end{enumerate}
\end{fact}

%
%
%
%

\subsection{Measures of Distance Between Distributions}\label{sec:prob-distance}

We also use the following standard measures of distance (or divergence) between distributions.

\paragraph{KL-divergence.} For two distributions $\mu$ and $\nu$, the \emph{Kullback-Leibler divergence} between $\mu$ and $\nu$ is denoted by $\kl{\mu}{\nu}$ and defined as:
\[
	\kl{\mu}{\nu} = \Ex_{a \sim \mu}\bracket{\log\frac{\mu(a)}{\nu(a)}}.
\]
The \emph{conditional KL-divergence} $\kl{\mu(\rA \mid \rB)}{\nu(\rA \mid \rB)}$ between two conditional distributions $\mu(\rA \mid \rB)$ and $\nu(\rA \mid \rB)$ is defined to be:
\[
	\kl{\mu(\rA \mid \rB)}{\nu(\rA \mid \rB)} = \Exp_{b \sim \mu(\rB)}\bracket{\kl{\mu(\rA \mid \rB=b)}{\nu(\rA \mid \rB=b)}}.
\]

\noindent
We use the following basic properties of KL-divergence.

\begin{fact}\label{fact:kl-event}
	Suppose $\mu$ is a distribution and $\event$ is an event, then,
	\[
		\kl{\mu \mid \event}{\mu} \le \log\paren{\frac{1}{\mu(\event)}}.
	\]
\end{fact}

\noindent
The following states the relation between mutual information and KL-divergence.

\begin{fact}\label{fact:kl-info}
	For random variables $\rA,\rB,\rC$,
	\[\mi{\rA}{\rB \mid \rC} = \Ex_{(b,c) \sim {(\rB,\rC)}}\bracket{ \kl{\distribution{\rA \mid \rB=b,\rC=c}}{\distribution{\rA \mid \rC=c}}}.\]
\end{fact}

\noindent
We also use the following chain rule of KL-divergence.

\begin{fact}\label{fact:kl-chain-rule}
	Suppose $\mu$ and $\nu$ are two distributions for $\rA,\rB$, then,
	\[
		\kl{\mu(\rA,\rB)}{\nu(\rA,\rB)} = \kl{\mu(\rA)}{\nu(\rA)} + \kl{\mu(\rB \mid \rA)}{\nu(\rB \mid \rA)}.
	\]
\end{fact}

\paragraph{Total variation distance.} We denote the total variation distance between two distributions $\mu$ and $\nu$ on the same support $\Omega$ by $\tvd{\mu}{\nu}$, defined as:
\[
	\tvd{\mu}{\nu} = \max_{\Omega' \subseteq \Omega} \paren{\mu(\Omega')-\nu(\Omega')} = \frac{1}{2} \cdot \sum_{x \in \Omega} \card{\mu(x) - \nu(x)}.
\]

\noindent
We use the following basic properties of total variation distance.

\begin{fact}\label{fact:tvd-small}
	Suppose $\mu$ and $\nu$ are two distributions for a non-negative random variable $\rA$, then,
	\[
		\card{\Exp_\mu\bracket{\rA} - \Exp_\nu\bracket{\rA}} \le \tvd{\mu}{\nu} \cdot \max_{a \in \supp{\rA}} a.
	\]
\end{fact}

\noindent
The total variation distance between two distributions can be bounded in terms of their KL-divergence by Pinsker's inequality.

\begin{fact}[Pinsker's inequality]\label{fact:pinskers}
	For distributions $\mu$ and $\nu$,
	\[
		\tvd{\mu}{\nu} \leq \sqrt{\frac{1}{2} \cdot \kl{\mu}{\nu}}.
	\]
\end{fact}

\noindent
An alternative bound for Pinsker's inequality is used for distributions with large KL-divergence.

\begin{fact}[\cite{Tsybakov09}, Equation 2.25]\label{fact:alt-pinskers}
	For distributions $\mu$ and $\nu$,
	\[
		\tvd{\mu}{\nu} \le 1 - \frac{1}{2} \cdot \exp(-\kl{\mu}{\nu}).
	\]
\end{fact}

\noindent
We also use the following chain rule of total variation distance.

\begin{fact}\label{fact:tvd-chain-rule}
	Suppose $\mu$ and $\nu$ are two distributions for $\rA,\rB$, then,
	\[
		\tvd{\mu(\rA,\rB)}{\nu(\rA,\rB)} \le \tvd{\mu(\rA)}{\nu(\rA)} + \Exp_{A \sim \mu(\rA)}\bracket{\tvd{\mu(\rB \mid \rA=A)}{\nu(\rB \mid \rA=A)}}.
	\]
\end{fact}



\section{Missing Proofs in~\Cref{sec:optimal}}\label{sec:proof}

We denote the \emph{bias} of a Boolean random variable $\rA$ by $\bias(\rA)$, which is defined as
\[
	\bias(\rA) = \card{\Pr(\rA=0) - \Pr(\rA=1)}.
\]
We use the following basic property regarding the biases of independent Boolean random variables.

\begin{fact}\label{fact:bias-xor}
	For independent Boolean random variables $\rA,\rB$,
	\[
		\bias(\rA \oplus \rB) = \bias(\rA) \cdot \bias(\rB).
	\]
\end{fact}

\subsection{Proof of~\Cref{lem:xor-direct-sum}}

We use (simplified versions of) Theorem 5.1 in~\cite{BarakBCR13}\footnote{Technically, Theorem 5.1 of~\cite{BarakBCR13} is for unbounded-round protocols. However, the bounded-round version also holds since the simulation protocol in their proof is round-preserving.} and Lemma 3.4 in~\cite{JainPY12}.

\begin{lemma}[\cite{BarakBCR13}, Theorem 5.1]\label{lem:bbcr}
	For $k,r \ge 1$, $\epsilon \in [0,1]$, Boolean function $f$, and input distribution $\mu$, there exists an $r$-round protocol $\pi$ for solving $f$ with probability $\epsilon$ over $\mu$ that has information cost $\ic_\mu(\pi) \le \bD^{(r)}_{\mu^k,\epsilon}(f^{\oplus k})/k+2$.
\end{lemma}

\begin{lemma}[\cite{JainPY12}, Lemma 3.4]\label{lem:jpy}
	For $r \ge 1$, $\epsilon \in (0,1)$, and input distribution $\mu$, any $r$-round protocol $\pi$ can be $\epsilon$-simulated (in $r$ rounds) with communication $O(r/\epsilon \cdot \ic_\mu(\pi) + r^2/\epsilon)$.
\end{lemma}

\begin{proof}[Proof of~\Cref{lem:xor-direct-sum}]
	By~\Cref{lem:bbcr}, there exists an $r$-round protocol $\pi$ for solving $f$ with probability $\epsilon$ over $\mu$ that has information cost $\ic_\mu(\pi) \le \bD^{(r)}_{\mu^k,\epsilon}(f^{\oplus k})/k+2$.
	Furthermore, $\pi$ can be $\delta$-simulated with communication $O(r/(\delta k) \cdot \bD^{(r)}_{\mu^k,\epsilon}(f^{\oplus k})+r^2/\delta)$ by~\Cref{lem:jpy}.
	The same bound holds for $\bD^{(r)}_{\mu,\epsilon-\delta}(f)$.
	Rearranging the terms concludes the proof.
\end{proof}

\subsection{Proof of~\Cref{lem:multi-xor}}

This section constitutes a proof of~\Cref{lem:multi-xor}.
Fix an $r$-round, $2k$-party protocol $\pi$ for solving $f^{\oplus k}$ over $\mu^k$ that has communication $C$.
Suppose for the sake of contradiction that $C < \bD^{(r)}_{\mu,1/2+\epsilon}(f)$.
In the following, we show that $\pi$ succeeds with probability less than $1/2+\min(\epsilon_1,\epsilon_2)$ over $\mu^k$, hence proving the lemma.
Without loss of generality, assume $\pi$ is deterministic.

We first introduce some random variables with respect to $\pi$ when input is drawn from $\mu^k$.
\begin{itemize}
	\item $\rX_i$: The input to Alice $i$ for $i \in [k]$;
	\item $\rY_i$: The input to Bob $i$ for $i \in [k]$;
	\item $\rM^{(j)}_i$: The $\ceil{j/2}$-th message posted by Alice $i$ if $j$ is odd, and the $\ceil{j/2}$-th message posted by Bob $i$ if $j$ is even, for $i \in [k]$ and $j \in [r+1]$;
	\item $\rB^{(j)}_i$: The blackboard after $\rM^{(j)}_i$ is posted for $i \in [k]$ and $j \in [r+1]$.
\end{itemize}
For convenience, we slightly abuse the notation and write $\rB^{(0)}_k = \perp$ and $\rB^{(j)}_0 = \rB^{(j-1)}_k$ for $j \in [r+1]$.
Let $\rB = \rB^{(r+1)}_k$.\footnote{Recall that for consistency with the standard $2$-party model, the multi-party model is defined so that the last party, who returns an output, does not post a message to the blackboard. Thus, there is an ``$(r+1)$-th round'' for an $r$-round protocol. To minimize confusion, we avoid any reference to the ``$j$-th round'' throughout.}
For blackboard $B$, $\bias(B)$ is defined to be $\bias(f^{\oplus k}(\rX,\rY) \mid \rB = B)$, and for $i \in [k]$, define $\bias_i(B)$ as $\bias(f(\rX_i,\rY_i) \mid \rB = B)$.
A couple of events are considered as well.
\begin{itemize}
	\item $\event_1(i,B)$: The event $\bias_i(B) \ge \sqrt{\epsilon}$ for $i \in [k]$ and blackboard $B$;
	\item $\event_1(S,B)$: The event $\bigwedge_{i \in S} \event_1(i,B)$ for $S \subseteq [k]$ and blackboard $B$;
	\item $\event_1(B)$: The event that there exists $S \subseteq [k]$ of size $(1-1/(10r)) \cdot k$ such that $\event_1(S,B)$ holds, for blackboard $B$;
	\item $\event_2(i,B)$: The event $\bias_i(B) \ge \epsilon^{1/(2r)}$ for $i \in [k]$ and blackboard $B$;
	\item $\event_2(S,B)$: The event $\bigwedge_{i \in S} \event_2(i,B)$ for $S \subseteq [k]$ and blackboard $B$;
	\item $\event_2(B)$: The event that there exists $S \subseteq [k]$ of size $(1-\epsilon) \cdot k$ such that $\event_2(S,B)$ holds, for blackboard $B$.
\end{itemize}

We will use the following basic properties of the protocol $\pi$.
Specifically, \Cref{clm:rectangle} proves a rectangle property and~\Cref{clm:advantage} bounds the success probability of $\pi$ in terms of biases.

\begin{claim}\label{clm:rectangle}
	For $i \in [k]$, $j \in [r+1]$, and \emph{fixed} blackboard $\rB^{(j)}_i$, it holds that
	\[
		\distribution{\rX,\rY \mid \rB^{(j)}_i} = \bigtimes_{i' \in [k]} \distribution{\rX_{i'},\rY_{i'} \mid \rB^{(j)}_i}.
	\]
\end{claim}

\begin{proof}
	The claim can be equivalently stated as that for $i,i' \in [k]$ and $j \in [r+1]$, it holds that
	\[
		\mi{\rX_{i'},\rY_{i'}}{\rX_{-i'},\rY_{-i'} \mid \rB^{(j)}_i} = 0.
	\]
	To this end, fix $i,i' \in [k]$ and $j \in [r+1]$.
	Suppose $i' = i$.
	Observe that
	\begin{align*}
		\mi{\rX_{i'},\rY_{i'}}{\rX_{-i'},\rY_{-i'} \mid \rB^{(j)}_i}
		& \le \mi{\rM^{(j)}_i,\rX_i,\rY_i}{\rX_{-i},\rY_{-i} \mid \rB^{(j)}_{i-1}} \tag{as $\rB^{(j)}_i = (\rB^{(j)}_{i-1},\rM^{(j)}_i)$}\\
		& = \mi{\rX_i,\rY_i}{\rX_{-i},\rY_{-i} \mid \rB^{(j)}_{i-1}}. \tag{as $\rM^{(j)}_i$ is a function of $\rB^{(j)}_{i-1},\rX_i,\rY_i$}
	\end{align*}
	For $i' \ne i$, also observe that
	\begin{align*}
		\mi{\rX_{i'},\rY_{i'}}{\rX_{-i'},\rY_{-i'} \mid \rB^{(j)}_i}
		& \le \mi{\rX_{i'},\rY_{i'}}{\rM^{(j)}_i,\rX_{-i'},\rY_{-i'} \mid \rB^{(j)}_{i-1}} \tag{as $\rB^{(j)}_i = (\rB^{(j)}_{i-1},\rM^{(j)}_i)$}\\
		& = \mi{\rX_{i'},\rY_{i'}}{\rX_{-i'},\rY_{-i'} \mid \rB^{(j)}_{i-1}}. \tag{as $\rM^{(j)}_i$ is a function of $\rB^{(j)}_{i-1},\rX_{-i'},\rY_{-i'}$}
	\end{align*}
	In both cases, we get
	\[
		\mi{\rX_{i'},\rY_{i'}}{\rX_{-i'},\rY_{-i'} \mid \rB^{(j)}_i} \le \mi{\rX_{i'},\rY_{i'}}{\rX_{-i'},\rY_{-i'} \mid \rB^{(j)}_{i-1}}.
	\]
	Consequently, an induction in increasing order of $(j,i)$ implies that
	\[
		\mi{\rX_{i'},\rY_{i'}}{\rX_{-i'},\rY_{-i'} \mid \rB^{(j)}_i} \le \mi{\rX_{i'},\rY_{i'}}{\rX_{-i'},\rY_{-i'}}.
	\]
	The claim follows as $\rX_{i'},\rY_{i'} \perp \rX_{-i'},\rY_{-i'}$.
\end{proof}

\begin{claim}\label{clm:advantage}
	The protocol $\pi$ succeeds with probability at most
	\[
		\frac{1}{2} + \frac{1}{2} \cdot \Exp_\rB\bracket{\prod_{i \in [k]} \bias_i(\rB)}.
	\]
\end{claim}

\begin{proof}
	Observe that given blackboard $B$, since the protocol $\pi$ is deterministic, it succeeds with probability at most
	\[
		\max_{b \in \set{0,1}} \Pr(f^{\oplus k}(\rX,\rY)=b \mid \rB=B) = \frac{1}{2} + \frac{1}{2} \cdot \bias(B).
	\]
	Summing over $B$, we get that the success probability of $\pi$ is upper bounded by
	\begin{align*}
		\frac{1}{2} + \frac{1}{2} \cdot \Exp_\rB\bracket{\bias(\rB)}
		& = \frac{1}{2} + \frac{1}{2} \cdot \Exp_\rB\bracket{\bias(f^{\oplus k}(\rX,\rY) \mid \rB)}\\
		& = \frac{1}{2} + \frac{1}{2} \cdot \Exp_\rB\bracket{\bias(\bigoplus_{i \in [k]} f(\rX_i,\rY_i) \mid \rB)}\\
		& = \frac{1}{2} + \frac{1}{2} \cdot \Exp_\rB\bracket{\prod_{i \in [k]} \bias(f(\rX_i,\rY_i) \mid \rB)} \tag{by~\Cref{fact:bias-xor} and~\Cref{clm:rectangle}}\\
		& = \frac{1}{2} + \frac{1}{2} \cdot \Exp_\rB\bracket{\prod_{i \in [k]} \bias_i(\rB)},
	\end{align*}
	as claimed.
\end{proof}

\begin{claim}\label{clm:small-prob-event}
	For $S \subseteq [k]$, it holds that
	\[
		\Pr_\rB(\event_1(S,\rB)) < (4\sqrt{\epsilon})^{\frac{|S|}{r+2}},
	\]
	and
	\[
		\Pr_\rB(\event_2(\rB)) < (6\epsilon^{1-\frac{1}{2r}})^{\frac{k}{r+2}}.
	\]
\end{claim}

The proof of~\Cref{lem:multi-xor} is done with the help of the above technical claim.
Assume its correctness for now.
We show that it indeed implies~\Cref{lem:multi-xor}.
Let $\cB_1$ be the subset of blackboards $B$ such that $\event_1(B)$ holds.
By~\Cref{clm:advantage}, the success probability of $\pi$ is at most
\begin{align*}
	& \frac{1}{2}+\frac{1}{2} \cdot \Exp_\rB\bracket{\prod_{i \in [k]} \bias_i(\rB)}\\
	& \hspace{1cm} = \frac{1}{2}+\frac{1}{2} \cdot \bracket{\sum_{B \in \cB_1} \Pr(\rB=B) \cdot \prod_{i \in [k]} \bias_i(B) + \sum_{B \not\in \cB_1} \Pr(\rB=B) \cdot \prod_{i \in [k]} \bias_i(B)}\\
	& \hspace{1cm} \le \frac{1}{2}+\frac{1}{2} \cdot \bracket{\Pr_\rB(\event_1(\rB)) + \sum_{B \not\in \cB_1} \Pr(\rB=B) \cdot \prod_{i \in [k]} \bias_i(B)} \tag{as $\bias_i(B) \in [0,1]$}\\
	& \hspace{1cm} \le \frac{1}{2}+\frac{1}{2} \cdot \bracket{\Pr_\rB(\event_1(\rB)) + (\sqrt{\epsilon})^{\frac{k}{10r}}}, \tag{as $\Pr_\rB(\overline{\event_1(\rB)}) \le 1$}
\end{align*}
where we use the observation that for $B \notin \cB_1$, $\bias_i(B) < \sqrt{\epsilon}$ for more than $k/(10r)$ indices $i \in [k]$.
By a union bound, we further upper bound the success probability of $\pi$ by
\begin{align*}
	\frac{1}{2}+\frac{1}{2} \cdot \Exp_\rB\bracket{\prod_{i \in [k]} \bias_i(\rB)}
	& \le \frac{1}{2}+\frac{1}{2} \cdot \bracket{\sum_{\substack{S \subseteq [k] : \\ \card{S}=(1-\frac{1}{10r}) \cdot k}}\Pr_\rB(\event_1(S,\rB)) + \epsilon^{\frac{k}{20r}}}\\
	& < \frac{1}{2}+\frac{1}{2} \cdot \bracket{\binom{k}{(1-\frac{1}{10r}) \cdot k} \cdot (4\sqrt{\epsilon})^{(1-\frac{1}{10r}) \cdot k \cdot \frac{1}{r+2}} + \epsilon^{\frac{k}{20r}}} \tag{by~\Cref{clm:small-prob-event}}\\
	& = \frac{1}{2}+\frac{1}{2} \cdot \bracket{\binom{k}{\frac{k}{10r}} \cdot (16\epsilon)^{\frac{10r-1}{20r} \cdot \frac{k}{r+2}} + \epsilon^{\frac{k}{20r}}}\\
	& \le \frac{1}{2}+\frac{1}{2} \cdot \bracket{\paren{\frac{ek}{k/(10r)}}^{\frac{k}{10r}} \cdot (16\epsilon)^{\frac{10r-1}{20r} \cdot \frac{k}{r+2}} + \epsilon^{\frac{k}{20r}}}\\
	& \le \frac{1}{2}+\epsilon_1,
\end{align*}
for sufficiently small $\epsilon < \epsilon_0$.

Meanwhile, let $\cB_2$ be the subset of blackboards $B$ such that $\event_2(B)$ holds.
Similarly by~\Cref{clm:advantage}, we can also upper bound the success probability of $\pi$ by
\[
	\frac{1}{2}+\frac{1}{2} \cdot \Exp_\rB\bracket{\prod_{i \in [k]} \bias_i(\rB)} \le \frac{1}{2}+\frac{1}{2} \cdot \bracket{\Pr_\rB(\event_2(\rB)) + (\epsilon^{\frac{1}{2r}})^{\epsilon k}},
\]
since for $B \not\in \cB_2$, $\bias_i(B) < \epsilon^{1/(2r)}$ for more than $\epsilon k$ indices $i \in [k]$.
By~\Cref{clm:small-prob-event}, the success probability of $\pi$ is at most
\[
	\frac{1}{2}+\frac{1}{2} \cdot \Exp_\rB\bracket{\prod_{i \in [k]} \bias_i(\rB)} < \frac{1}{2}+\frac{1}{2} \cdot \bracket{(6\epsilon^{1-\frac{1}{2r}})^{\frac{k}{r+2}} + \epsilon^{\frac{\epsilon k}{2r}}} \le \frac{1}{2}+\epsilon_2,
\]
for sufficiently small $\epsilon < \epsilon_0$.
This concludes the proof of~\Cref{lem:multi-xor}.
It now remains to show the correctness of~\Cref{clm:small-prob-event}.
Suppose not, we will construct an $r$-round, $2$-party protocol $\tau$ for solving $f$ with probability $1/2+\epsilon$ over $\mu$ that has communication $C$, contradicting the assumption $C < \bD^{(r)}_{\mu,1/2+\epsilon}(f)$.
On input $(X,Y)$, $\tau$ will simulate $\pi$ as shown in~\Cref{alg:xor}, where all random variables are with respect to $\pi$.
The parameters $T \subseteq [k]$ and event $\event$ are to be determined and are supposed to be such that $\bias_\rI(\rB)$ is large with high probability, conditioned on $\event$, for $\rI \sim T$.

\begin{Algorithm}\label{alg:xor}
	The protocol $\tau$ for solving $f$ on input $(X,Y)$.

	\textbf{Parameters:} $T \subseteq [k]$ and event $\event$.
	\begin{enumerate}
		\item Alice and Bob publicly sample $\rI$ uniformly at random from $T$.
		\item Alice and Bob publicly sample $\rX_{< \rI},\rY_{< \rI}$ conditioned on $\rI,\event$.
		\item Alice sets $\rX_\rI = X$ and Bob sets $\rY_\rI = Y$.
		\item For $j \in [r+1]$, if $j$ is odd,
			\begin{enumerate}
				\item Alice simulates $\pi$ on $\rX_{\le \rI}$ and computes $\rM^{(j)}_{\le \rI}$, given $\rB^{(j-1)}_k$.
				\item Alice privately samples $\rM^{(j)}_{> \rI}$ conditioned on $\rB^{(j)}_\rI,\rX_{< \rI},\rY_{< \rI},\rI,\event$.
				\item Alice sends $\rM^{(j)}$ to Bob if $j \le r$, and otherwise outputs
					\[
						\argmax_{b \in \set{0,1}} \Pr_{(x,y) \sim \mu^k \mid \rB}(f(x_i,y_i) = b);
					\]
			\end{enumerate}
			if $j$ is even,
			\begin{enumerate}
				\item Bob simulates $\pi$ on $\rY_{\le \rI}$ and computes $\rM^{(j)}_{\le \rI}$, given $\rB^{(j-1)}_k$.
				\item Bob privately samples $\rM^{(j)}_{> \rI}$ conditioned on $\rB^{(j)}_\rI,\rX_{< \rI},\rY_{< \rI},\rI,\event$.
				\item Bob sends $\rM^{(j)}$ to Alice if $j \le r$, and otherwise outputs
					\[
						\argmax_{b \in \set{0,1}} \Pr_{(x,y) \sim \mu^k \mid \rB}(f(x_i,y_i) = b).
					\]
			\end{enumerate}
	\end{enumerate}
\end{Algorithm}

It can be verified that the current transcript always reflects the up-to-date blackboard and thus $\tau$ is indeed an $r$-round protocol for solving $f$ that has communication $C$.
Regarding the success probability of $\tau$, the following two technical claims show that all the random variables as sampled in $\tau$ almost perfectly follow their distribution in $\pi$ conditioned on $\event$.

\begin{claim}\label{clm:small-kl-input}
	For $T \subseteq [k]$ and event $\event$, it holds that
	\[
		\Exp_{\rX_{< \rI},\rY_{< \rI},\rI \mid \event}\kl{\distribution{\rX_\rI,\rY_\rI \mid \rX_{< \rI},\rY_{< \rI},\rI,\event}}{\distribution{\rX_\rI,\rY_\rI}} \le \frac{1}{|T|} \cdot \log\frac{1}{\Pr(\event)}.
	\]
\end{claim}

\begin{proof}
	Observe that
	\begin{align*}
		& \Exp_{\rX_{< \rI},\rY_{< \rI},\rI \mid \event}\kl{\distribution{\rX_\rI,\rY_\rI \mid \rX_{< \rI},\rY_{< \rI},\rI,\event}}{\distribution{\rX_\rI,\rY_\rI}}\\
		& \hspace{1cm} = \kl{\distribution{\rX_\rI,\rY_\rI \mid \rX_{< \rI},\rY_{< \rI},\rI,\event}}{\distribution{\rX_\rI,\rY_\rI \mid \rX_{< \rI},\rY_{< \rI},\rI}} \tag{as $\rX_\rI,\rY_\rI \perp \rX_{< \rI},\rY_{< \rI},\rI$}\\
		& \hspace{1cm} = \frac{1}{|T|} \cdot \sum_{i \in T} \kl{\distribution{\rX_i,\rY_i \mid \rX_{< i},\rY_{< i},\rI=i,\event}}{\distribution{\rX_i,\rY_i \mid \rX_{< i},\rY_{< i},\rI=i}} \tag{as $\rI$ is uniform over $T$}\\
		& \hspace{1cm} = \frac{1}{|T|} \cdot \sum_{i \in T} \kl{\distribution{\rX_i,\rY_i \mid \rX_{< i},\rY_{< i},\event}}{\distribution{\rX_i,\rY_i \mid \rX_{< i},\rY_{< i}}} \tag{as $\rX_{\le i},\rY_{\le i} \perp \rI=i \mid \event$ and $\rX_{\le i},\rY_{\le i} \perp \rI=i$}\\
		& \hspace{1cm} \le \frac{1}{|T|} \cdot \sum_{i \in [k]} \kl{\distribution{\rX_i,\rY_i \mid \rX_{< i},\rY_{< i},\event}}{\distribution{\rX_i,\rY_i \mid \rX_{< i},\rY_{< i}}} \tag{as $T \subseteq [k]$}\\
		& \hspace{1cm} = \frac{1}{|T|} \cdot \kl{\distribution{\rX,\rY \mid \event}}{\distribution{\rX,\rY}} \tag{by chain rule of KL-divergence (\Cref{fact:kl-chain-rule})}\\
		& \hspace{1cm} \le \frac{1}{|T|} \cdot \log\frac{1}{\Pr(\event)}, \tag{by~\Cref{fact:kl-event}}
	\end{align*}
	as claim.
\end{proof}

\begin{claim}\label{clm:small-kl-message}
	For $j \in [r+1]$, $T \subseteq [k]$, and event $\event$, it holds that
	\[
		\Exp_{\rB^{(j)}_\rI,\rX_{\le \rI},\rY_{\le \rI},\rI \mid \event}\kl{\distribution{\rM^{(j)}_{> \rI} \mid \rB^{(j)}_\rI,\rX_{\le \rI},\rY_{\le \rI},\rI,\event}}{\distribution{\rM^{(j)}_{> \rI} \mid \rB^{(j)}_\rI,\rX_{< \rI},\rY_{< \rI},\rI,\event}} \le \frac{1}{|T|} \cdot \log\frac{1}{\Pr(\event)}.
	\]
\end{claim}

\begin{proof}
	Observe that
	\begin{align*}
		& \Exp_{\rB^{(j)}_\rI,\rX_{\le \rI},\rY_{\le \rI},\rI \mid \event}\kl{\distribution{\rM^{(j)}_{> \rI} \mid \rB^{(j)}_\rI,\rX_{\le \rI},\rY_{\le \rI},\rI,\event}}{\distribution{\rM^{(j)}_{> \rI} \mid \rB^{(j)}_\rI,\rX_{< \rI},\rY_{< \rI},\rI,\event}}\\
		& \hspace{1cm} = \Exp_{\rB^{(j)}_k,\rX_{\le \rI},\rY_{\le \rI},\rI \mid \event} \log\frac{\Pr(\rM^{(j)}_{> \rI} \mid \rB^{(j)}_\rI,\rX_{\le \rI},\rY_{\le \rI},\rI,\event)}{\Pr(\rM^{(j)}_{> \rI} \mid \rB^{(j)}_\rI,\rX_{< \rI},\rY_{< \rI},\rI,\event)} \tag{as $\rB^{(j)}_k = (\rB^{(j)}_\rI,\rM^{(j)}_{> \rI})$}\\
		& \hspace{1cm} = \Exp_{\rB^{(j)}_k,\rX_{\le \rI},\rY_{\le \rI},\rI \mid \event} \log\frac{\Pr(\rX_\rI,\rY_\rI \mid \rB^{(j)}_k,\rX_{< \rI},\rY_{< \rI},\rI,\event)}{\Pr(\rX_\rI,\rY_\rI \mid \rB^{(j)}_\rI,\rX_{< \rI},\rY_{< \rI},\rI,\event)} \tag{as $\frac{\Pr(A \mid B)}{\Pr(A)} = \frac{\Pr(B \mid A)}{\Pr(B)}$}\\
		& \hspace{1cm} = \Exp_{\rB^{(j)}_k,\rX_{\le \rI},\rY_{\le \rI},\rI \mid \event} \log\frac{\Pr(\rX_\rI,\rY_\rI \mid \rB^{(j)}_k,\rX_{< \rI},\rY_{< \rI},\rI,\event)}{\Pr(\rX_\rI,\rY_\rI \mid \rB^{(j)}_k,\rX_{< \rI},\rY_{< \rI},\rI)}\\
		& \hspace{2cm} - \Exp_{\rB^{(j)}_k,\rX_{\le \rI},\rY_{\le \rI},\rI \mid \event} \log\frac{\Pr(\rX_\rI,\rY_\rI \mid \rB^{(j)}_\rI,\rX_{< \rI},\rY_{< \rI},\rI,\event)}{\Pr(\rX_\rI,\rY_\rI \mid \rB^{(j)}_k,\rX_{< \rI},\rY_{< \rI},\rI)} \tag{as $\log\frac{a}{b} = \log\frac{a}{c} - \log\frac{b}{c}$}\\
		& \hspace{1cm} = \kl{\distribution{\rX_\rI,\rY_\rI \mid \rB^{(j)}_k,\rX_{< \rI},\rY_{< \rI},\rI,\event}}{\distribution{\rX_\rI,\rY_\rI \mid \rB^{(j)}_k,\rX_{< \rI},\rY_{< \rI},\rI}}\\
		& \hspace{2cm} - \Exp_{\rB^{(j)}_k,\rX_{\le \rI},\rY_{\le \rI},\rI \mid \event} \log\frac{\Pr(\rX_\rI,\rY_\rI \mid \rB^{(j)}_\rI,\rX_{< \rI},\rY_{< \rI},\rI,\event)}{\Pr(\rX_\rI,\rY_\rI \mid \rB^{(j)}_\rI,\rX_{< \rI},\rY_{< \rI},\rI)} \tag{as $\rX_\rI,\rY_\rI \perp \rM^{(j)}_{> \rI} \mid \rB^{(j)}_\rI,\rX_{< \rI},\rY_{< \rI},\rI$ by~\Cref{clm:rectangle}}\\
		& \hspace{1cm} = \kl{\distribution{\rX_\rI,\rY_\rI \mid \rB^{(j)}_k,\rX_{< \rI},\rY_{< \rI},\rI,\event}}{\distribution{\rX_\rI,\rY_\rI \mid \rB^{(j)}_k,\rX_{< \rI},\rY_{< \rI},\rI}}\\
		& \hspace{2cm} - \kl{\distribution{\rX_\rI,\rY_\rI \mid \rB^{(j)}_\rI,\rX_{< \rI},\rY_{< \rI},\rI,\event}}{\distribution{\rX_\rI,\rY_\rI \mid \rB^{(j)}_\rI,\rX_{< \rI},\rY_{< \rI},\rI}}\\
		& \hspace{1cm} \le \frac{1}{|T|} \cdot \sum_{i \in T} \kl{\distribution{\rX_i,\rY_i \mid \rB^{(j)}_k,\rX_{< i},\rY_{< i},\rI=i,\event}}{\distribution{\rX_i,\rY_i \mid \rB^{(j)}_k,\rX_{< i},\rY_{< i},\rI=i}} \tag{as $\rI$ is uniform over $T$}\\
		& \hspace{1cm} = \frac{1}{|T|} \cdot \sum_{i \in T} \kl{\distribution{\rX_i,\rY_i \mid \rB^{(j)}_k,\rX_{< i},\rY_{< i},\event}}{\distribution{\rX_i,\rY_i \mid \rB^{(j)}_k,\rX_{< i},\rY_{< i}}} \tag{as $\rB^{(j)}_k,\rX_{\le i},\rY_{\le i} \perp \rI=i \mid \event$ and $\rB^{(j)}_k,\rX_{\le i},\rY_{\le i} \perp \rI=i$}\\
		& \hspace{1cm} \le \frac{1}{|T|} \cdot \sum_{i \in [k]} \kl{\distribution{\rX_i,\rY_i \mid \rB^{(j)}_k,\rX_{< i},\rY_{< i},\event}}{\distribution{\rX_i,\rY_i \mid \rB^{(j)}_k,\rX_{< i},\rY_{< i}}} \tag{as $T \subseteq [k]$}\\
		& \hspace{1cm} = \frac{1}{|T|} \cdot \kl{\distribution{\rX,\rY \mid \rB^{(j)}_k,\event}}{\distribution{\rX,\rY \mid \rB^{(j)}_k}} \tag{by chain rule of KL-divergence (\Cref{fact:kl-chain-rule})}\\
		& \hspace{1cm} \le \frac{1}{|T|} \cdot \Exp_{\rB^{(j)}_k \mid \event} \log\frac{1}{\Pr(\event \mid \rB^{(j)}_k)} \tag{by~\Cref{fact:kl-event}}\\
		& \hspace{1cm} \le \frac{1}{|T|} \cdot \log\Exp_{\rB^{(j)}_k \mid \event} \frac{1}{\Pr(\event \mid \rB^{(j)}_k)} \tag{by concavity of $\log(\cdot)$}\\
		& \hspace{1cm} = \frac{1}{|T|} \cdot \log\sum_{\rB^{(j)}_k} \frac{\Pr(\rB^{(j)}_k \mid \event)}{\Pr(\event \mid \rB^{(j)}_k)}\\
		& \hspace{1cm} = \frac{1}{|T|} \cdot \log\sum_{\rB^{(j)}_k} \frac{\Pr(\rB^{(j)}_k)}{\Pr(\event)} \tag{as $\frac{\Pr(A \mid B)}{\Pr(B \mid A)} = \frac{\Pr(A)}{\Pr(B)}$}\\
		& \hspace{1cm} = \frac{1}{|T|} \cdot \log\frac{1}{\Pr(\event)}.
	\end{align*}
	This concludes the proof.
\end{proof}

We emphasize that Alice is able to compute $\rM^{(j)}_{\le \rI}$ exactly for odd $j \in [r+1]$ as it is fully determined by $\rB^{(j-1)}_k$, which is provided by the current transcript, and $\rX_{\le \rI}$.
Similarly, Bob is able to compute $\rM^{(j)}_{\le \rI}$ exactly for even $j \in [r+1]$ since he has full knowledge of $\rB^{(j)}_k$ and $\rY_{\le \rI}$.
We are now ready to prove~\Cref{clm:small-prob-event}.

\begin{proof}[Proof of~\Cref{clm:small-prob-event}]
	Fix the input distribution $\mu$.
	Let $\nu_\pi$ be the distribution of $\rB,\rX_{\le \rI},\rY_{\le \rI},\rI$ where $\rI$ is drawn from $T$ uniformly at random and $\rB,\rX_{\le \rI},\rY_{\le \rI}$ follow their distribution in $\pi$ conditioned on $\event$.
	Also let $\nu_\tau$ be the distribution of $\rB,\rX_{\le \rI},\rY_{\le \rI},\rI$ as sampled in $\tau$.
	By chain rule of KL-divergence (\Cref{fact:kl-chain-rule}), together with~\Cref{clm:small-kl-input,clm:small-kl-message}, we get that
	\[
		\kl{\nu_\pi}{\nu_\tau} \le \frac{r+2}{|T|} \cdot \log\frac{1}{\Pr(\event)}.
	\]
	Using an alternative bound for Pinsker's inequality (\Cref{fact:alt-pinskers}), we further have
	\[
		\tvd{\nu_\pi}{\nu_\tau} \le 1 - \frac{1}{2} \cdot \Pr(\event)^{\frac{r+2}{|T|}}.
	\]
	Fix a threshold value $t \in [0,1]$.
	Observe that the success probability of $\tau$ is
	\begin{align*}
		\frac{1}{2} + \frac{1}{2} \cdot \Exp_{\rB,\rI \sim \nu_\tau}\bracket{\bias_\rI(\rB)}
		& \ge \frac{1}{2} + \frac{t}{2} \cdot \Pr_{\rB,\rI \sim \nu_\tau}\paren{\bias_\rI(\rB) \ge t} \tag{as $\bias_\rI(\rB) \in [0,1]$}\\
		& \ge \frac{1}{2} + \frac{t}{2} \cdot \bracket{\Pr_{\rB,\rI \sim \nu_\pi}\paren{\bias_\rI(\rB) \ge t} - \tvd{\nu_\pi}{\nu_\tau}} \tag{by~\Cref{fact:tvd-small}}\\
		& \ge \frac{1}{2} + \frac{t}{2} \cdot \bracket{\Pr_{\rB,\rI \sim \nu_\pi}\paren{\bias_\rI(\rB) \ge t} + \frac{1}{2} \cdot \Pr(\event)^{\frac{r+2}{|T|}} - 1}.
	\end{align*}
	Recall that $\tau$ has communication $C$ and thus by assumption, it can only succeed with probability less than $1/2+\epsilon$.
	Rearranging the terms above, we get
	\[
		\Pr(\event) < \bracket{2+\frac{4\epsilon}{t} - 2 \cdot \Pr_{\rB,\rI \sim \nu_\pi}\paren{\bias_\rI(\rB) \ge t}}^{\frac{|T|}{r+2}}.
	\]
	For $S \subseteq [k]$, setting $T=S$, $\event=\event_1(S,\rB)$, and $t=\sqrt{\epsilon}$ implies
	\begin{align*}
		\Pr_\rB(\event_1(S,\rB))
		& < \bracket{2+4\sqrt{\epsilon}-2 \cdot \Pr_{\substack{\rB \mid \event_1(S,\rB),\\ \rI \sim S}}\paren{\bias_\rI(\rB) \ge \sqrt{\epsilon}}}^{\frac{|S|}{r+2}}\\
		& = \paren{2+4\sqrt{\epsilon}-2 \cdot 1}^{\frac{|S|}{r+2}}\\
		& = (4\sqrt{\epsilon})^{\frac{|S|}{r+2}},
	\end{align*}
	since conditioned on $\event_1(S,\rB)$, it always holds that $\bias_\rI(\rB) \ge \sqrt{\epsilon}$ for $\rI \in S$.
	Meanwhile, setting $T=[k]$, $\event=\event_2(\rB)$, and $t=\epsilon^{1/(2r)}$ implies
	\begin{align*}
		\Pr_\rB(\event_2(\rB))
		& < \bracket{2+4\epsilon^{1-\frac{1}{2r}}-2 \cdot \Pr_{\substack{\rB \mid \event_2(\rB),\\ \rI \sim [k]}}\paren{\bias_\rI(\rB) \ge \epsilon^{\frac{1}{2r}}}}^{\frac{k}{r+2}}\\
		& \le \paren{2+4\epsilon^{1-\frac{1}{2r}}-2 \cdot (1-\epsilon)}^{\frac{k}{r+2}}\\
		& \le (6\epsilon^{1-\frac{1}{2r}})^{\frac{k}{r+2}}, \tag{as $\epsilon \le \epsilon^{1-\frac{1}{2r}}$}
	\end{align*}
	where in the second step, we use the fact that conditioned on $\event_2(\rB)$, at least a $1-\epsilon$ fraction of indices $\rI \in [k]$ satisfy $\bias_\rI(\rB) \ge \epsilon^{1/(2r)}$.
	This concludes the proof.
\end{proof}

\end{document}